\pdfoutput=1
\documentclass[a4paper,11pt]{article}
\usepackage{hyphenat}

\usepackage[utf8]{inputenc}
\usepackage[T1]{fontenc}
\usepackage{lmodern}
\usepackage{microtype}

\usepackage[numbers,sort&compress]{natbib}

\usepackage[margin=1in]{geometry}

\usepackage{mathtools} 
\usepackage{amssymb}
\usepackage{amsthm}

\usepackage{graphicx}
\usepackage{placeins} 
\DeclareGraphicsExtensions{.pdf,.png,.jpg}
\graphicspath{{./}{figs/}}

\usepackage{subcaption} 
\usepackage{hyperref}
\hypersetup{
colorlinks=true,
linkcolor=blue,
citecolor=blue,
urlcolor=blue,
pdfauthor={Tatsuaki Tsuruyama},
pdftitle={Kullback--Leibler Divergence Potential for Non-Ergodic Replication Dynamics: An Information-Theoretic Second Law}
}
\usepackage{doi} 

\newtheorem{lemma}{Lemma}
\newtheorem{theorem}{Theorem}
\newtheorem{proposition}{Proposition}
\newtheorem{remark}{Remark}

\newcommand{\figindent}{1cm}
\newcommand{\epsTwo}{0.073}
\newcommand{\delTwo}{0.83}
\newcommand{\epsFour}{0.051}
\newcommand{\delFour}{0.88}
\newcommand{\epsHuge}{0.019}
\newcommand{\delHuge}{0.93}
\newcommand{\blockstats}[2]{
\par\vspace{0.25em}
{\scriptsize\emph{Doeblin} $\varepsilon=#1$, \emph{Dobrushin} $\delta=#2$.}
}

\usepackage{authblk}
\title{Kullback--Leibler Divergence Potential for Non-Ergodic Replication Dynamics: An Information-Theoretic Second Law}

\author[1,2,3]{Tatsuaki Tsuruyama\thanks{E-mail: \href{mailto:tsuruyam@kuhp.kyoto-u.ac.jp}{tsuruyam@kuhp.kyoto-u.ac.jp}}}
\affil[1]{Department of Physics, Tohoku University, Sendai 980-8578, Japan}
\affil[2]{Department of Drug Discovery Medicine, Kyoto University, Kyoto 606-8501, Japan}
\affil[3]{Department of Clinical Laboratory, Kyoto Tachibana University, Kyoto 607-8175, Japan}

\date{} 

\begin{document}
\maketitle
\begin{abstract}
While previous research in information thermodynamics has focused on the thermodynamic costs associated with information ``erasure'' or ``measurement'' through concepts such as Landauer’s principle and mutual information, little systematic discussion has addressed the inherently irreversible nature of ``replication'' itself and the accompanying degradation of information structure. In this study, we construct a mathematical model of information replication using a discrete Markov model and Gaussian convolution, and quantify changes in information at each replication step: Shannon entropy, cross-entropy, and the Kullback--Leibler divergence (KLD). The monotonic decrease of the \emph{KLD potential} $V$ (the minimal KLD to a reachable steady set) exhibits a Lyapunov-like property, which can be interpreted as a potential analogous to the free energy in the process by which a nonequilibrium system converges to a particular steady state. Furthermore, we extend this framework to the potential applicability to biological information processes such as DNA replication, showing that the free energy required for degradation and repair can be expressed in terms of KLD. This contributes to building a unified information-thermodynamic framework for operations such as replication, transmission, and repair of information.
\end{abstract}

\medskip

\section{General framework and main theorem}\label{sec:framework}

\subsection{Block invariance and reachable steady set}\label{subsec:block-invariance}

The information replication operations we study (image blurring, base substitution with proofreading, etc.) are inherently \emph{local}: one step only rearranges probability mass within a constrained neighborhood (spatial or compositional), but does not freely mix all microscopic states.  As a result, certain coarse observables, e.g., the total mass inside a spatial region of an image, or the adenine(A) -thymine (T) rich block and guanine(G) -cytosine (C) rich block in DNA to be replicated as a template, are effectively \emph{conserved} by every replication step.  Empirical mutation processes are asymmetric between the AT and GC blocks. First, in various bacteria, the spectrum of spontaneous mutation is generally biased toward AT (C/G$\to$T/A transitions dominate), implying a baseline GC$\to$AT pressure \cite{HershbergPetrov2010}.  Second, cytosines in CpG dinucleotides (often methylated as 5mC) are hypermutable via deamination, which elevates C$\to$T transitions and locally accelerates GC$\to$AT decay \cite{FryxellMoon2005,Cooper2010}.  Third, germline substitution rates vary strongly with the local sequence context; in humans, the 7-mer sequence context explains most of the variability in polymorphism levels \cite{AggarwalaVoight2016}.  In contrast to AT-directed mutational bias, recombination-associated GC-biased gene conversion (gBGC) favors the fixation of GC alleles and can increase GC content in high recombination regions without invoking fitness differences \cite{DuretGaltier2009,CapraPollard2013,Weber2014}.

These well-documented asymmetries motivate modeling replication on a block partition that separates AT- from GC-rich components. 
In our framework, block invariance captures slow exchange between AT and GC pools while allowing within-block relaxation; consequently, the KLD potential $V$ measures the departure from the stable block set under realistic composition-dependent mutation / repair biases. When such conserved coarse variables exist, the Markov dynamics does not converge to a \emph{single} global invariant distribution; instead, the state space decomposes into components that evolve independently.  We call this situation \emph{non-ergodic} in the sense that long-term behavior retains memory of the initial coarse structure through the conserved weights.

(i) In the image model, a blockwise Gaussian convolution forbids smoothing across block boundaries; the total intensity in each block is preserved at every step.  
(ii) In the DNA model, a block diagonal substitution kernel reflects biochemical or sequence constraints (e.g., AT- vs. GC-rich segments); block masses (AT vs. GC) remain fixed while within-block compositions relax.  
(iii) More generally, locality, topological separation, or routing restrictions in copying channels induce invariant 'blocks' that prevent global mixing.
Because replication preserves these coarse variables, the appropriate equilibrium notion is not a single point but a \emph{set} of steady states obtained by fully relaxing \emph{within} each invariant component while keeping the initial coarse weights fixed.  This motivates modeling the state space as a disjoint union of blocks and defining a \emph{reachable steady set} against which we will measure the distance using a KL-based potential.
We emphasise that non-trivial block partitions
$\mathcal{X} = \bigsqcup_{j=1}^m \mathcal{X}_j$ with $m>1$
are meaningful precisely when the dynamics itself possesses non-mixing components,
so that probability mass cannot freely move between all microscopic states.
In the fully mixing (ergodic) case, the only natural choice is the trivial
partition with a single block, and our potential $V$ reduces to the standard
KLD to a unique invariant steady state.

\subsection*{Framework and notation (matrix semantics made explicit)}

We consider a finite state space partitioned into $m\in\mathbb N$ disjoint blocks
\begin{equation}
  \mathcal{X}\ =\ \bigsqcup_{j=1}^{m}\,\mathcal{X}_j,
  \label{eq:state-partition}
\end{equation}
with $|\mathcal X|=d<\infty$. The probability simplex over $\mathcal X$ is
\[
\Delta(\mathcal X)\ :=\ \bigl\{\,p\in\mathbb R^{d}\ :\ p\ge 0,\ \mathbf 1^\top p=1\,\bigr\},
\]
and we write $p=(p_x)_{x\in\mathcal X}\in\Delta(\mathcal X)$.

A one-step Markov kernel $T$ acts on the right by
\[
(pT)(y)=\sum_{x\in\mathcal X} p(x)\,T(x,y),\qquad p_{n+1}=p_n\,T,
\]
so $T$ is a $d\times d$ \emph{row-stochastic} matrix ($T(x,y)\ge0$ and
$\sum_{y}T(x,y)=1$ for each row $x$). The \emph{matrix dimension} $d\times d$ equals the number of states $|\mathcal X|$.
Rows index the \emph{current} state, columns the \emph{next} state. Each block kernel $T_j$ has size $|\mathcal X_j|\times|\mathcal X_j|$ and updates the conditional $p^{(j)}$ while preserving its mass $w_j(p_0)$. This row-vector/right-action convention
is used throughout: $p_{n+1}=p_n\,T$.

\emph{Semantics:} the \textbf{row index} is the current state, the \textbf{column index}
is the next state, and the \textbf{matrix size} $d\times d$ equals the number of states.
Restricting to a block $\mathcal X_j$ yields a subkernel $T_j$ of size
$|\mathcal X_j|\times|\mathcal X_j|$.

\paragraph{Block-invariant dynamics.}
The dynamics is \emph{block-invariant} if the transitions never cross the block boundaries:
\begin{equation}
  x\in\mathcal{X}_j\ \Longrightarrow\ T(x,\mathcal{X}_j)=1\qquad (j=1,\dots,m),
  \label{eq:block-invariance}
\end{equation}
equivalently $T(x,y)=0$ for $y\notin\mathcal X_j$ whenever $x\in\mathcal X_j$.
Matrix-wise, $T$ is block diagonal: $T=\mathrm{diag}(T_1,\dots,T_m)$.
Under \eqref{eq:block-invariance}, \emph{block masses} are conserved and conditionals evolve
block-wise:
\begin{equation}
  w_j(p_{n+1})=w_j(p_n),\qquad (Tp)^{(j)}=T_j\,p^{(j)}.
  \label{eq:mass-cons-cond-evol}
\end{equation}

\paragraph{Coarse variables (block masses) and within-block conditionals.}
For $p\in\mathcal P(\mathcal X)$ define
\begin{equation}
  w_j(p):=p(\mathcal X_j)=\sum_{x\in\mathcal X_j}p(x),\qquad
  p^{(j)}(x):=\begin{cases} p(x)/w_j(p), & x\in\mathcal X_j,\\ 0, & x\notin\mathcal X_j.\end{cases}
  \label{eq:block-mass-cond}
\end{equation}
Then $p$ decomposes as a convex mixture of conditionals within the block:
\begin{equation}
  p\ =\ \sum_{j=1}^m w_j(p)\,p^{(j)}.
  \label{eq:mixture-decomp}
\end{equation}

For each block, invariant measures and nonemptiness of the steady set is given by:
\begin{equation}
  \mathcal{I}_j\ :=\ \bigl\{\ \pi\in\mathcal{P}(\mathcal{X}_j)\ :\ \pi T_j=\pi\ \bigr\}.
  \label{eq:Ij}
\end{equation}

\emph{Existence.} Since every $T_j$ is a finite row-stochastic matrix, it admits at least one
invariant distribution (e.g.\ by Brouwer’s fixed point theorem on the simplex or
Perron–Frobenius for $T_j^\top$), hence $\mathcal I_j\neq\emptyset$ for all $j$. Given an initial distribution $p_0$, 
the \emph{reachable steady set} is
\label{eq:V-def-intro}

  \label{eq:V-def-intro2}
  \label{eq:V-def-intro-b}
  \[ \Pi(p_0) := \Bigl\{ \pi = \sum_{j=1}^{m} w_j(p_0)\,\pi_j : \pi_j \in \mathcal{I}_j \Bigr\}, \]
  \label{eq:Pi-def}

which is nonempty. Intuitively, $\Pi(p_0)$ collects all steady states that (i) fully relax
\emph{within} each block and (ii) preserve the initial coarse masses $w_j(p_0)$ fixed by
non-ergodic constraints.

\paragraph*{Block primitivity.}
\emph{Primitivity} of a block kernel $T_j$ (irreducible and aperiodic; equivalently,
$(T_j)^n>0$ for some $n$) is invoked \emph{only} to conclude that the within-block
conditionals converge to a \emph{unique} invariant $\pi_j^\ast\in\mathcal I_j$. 
For a block kernel $K$, with constant $\varepsilon$, means
\begin{equation}
  \exists\,\nu\in\Delta(\mathcal X_j),\ \varepsilon>0\ \text{ s.t. }\ 
  K(x,\cdot)\ \ge\ \varepsilon\,\nu(\cdot)\quad \forall x\in\mathcal X_j,
  \label{eq:Doeblin-cond}
\end{equation}
which implies total-variation contraction at most $1-\varepsilon$ per step.
We also define coefficient
\begin{equation}
  \delta(K)\ :=\ 1-\min_{x,x'}\sum_{y\in\mathcal X_j}\min\{K(x,y),K(x',y)\},
  \label{eq:Dobrushin-def}
\end{equation}
which bounds the one-step total-variation contraction by $\delta(K)$.

\subsection{Key Lemmas}\label{subsec:lemmas}
\begin{lemma}[Preservation of block mass and conditional evolution]\label{lem:block-preservation}
Let $p_{+}=Tp$ with block-invariant $T$ as in \eqref{eq:block-invariance}. 
Then for all $j=1,\dots,m$,
\begin{equation}
  w_j(p_{+})\ =\ w_j(p),
  \label{eq:block-mass-preserve}
\end{equation}
and, whenever $w_j(p)>0$,
\begin{equation}
  (Tp)^{(j)}\ =\ T_j\,p^{(j)}.
  \label{eq:conditional-evolution}
\end{equation}
\end{lemma}

\begin{proof}
By block invariance, $T(x,\mathcal{X}_j)=1$ for $x\in\mathcal{X}_j$ and $0$ otherwise. Hence
\( w_j(p_{+})=(Tp)(\mathcal{X}_j)=\sum_{x\in\mathcal{X}}p(x)\,T(x,\mathcal{X}_j)
= \sum_{x\in\mathcal{X}_j}p(x)\cdot1+\sum_{x\notin\mathcal{X}_j}p(x)\cdot0
= \sum_{x\in\mathcal{X}_j}p(x)=w_j(p) \).

which proves \eqref{eq:block-mass-preserve}. For any measurable $A\subseteq\mathcal{X}_j$,

\begin{align}
(Tp)^{(j)}(A)
&= \frac{(Tp)(A)}{w_j(p)}                                   \notag\\
&= \frac{\sum_{x\in\mathcal{X}_j} p(x)\,T_j(x,A)}{w_j(p)}    \notag\\
&= \sum_{x\in\mathcal{X}_j} \frac{p(x)}{w_j(p)}\,T_j(x,A)    \notag\\
&= (T_j\,p^{(j)})(A).
\label{eq:conditional-evolution-proof}
\end{align}

establishing \eqref{eq:conditional-evolution}.
\end{proof}

\begin{lemma}[Block decomposition of KLD]\label{lem:block-decomposition}
Let $\pi=\sum_{j=1}^{m} w_j\,\pi_j$ with $w_j>0$ and $\pi_j\in\mathcal{P}(\mathcal{X}_j)$ supported in disjoint blocks $\mathcal{X}_j$. 
Then, for any $p\in\mathcal{P}(\mathcal{X})$,
\begin{equation}
  D_{\mathrm{KL}}(p\|\pi)\ 
  =\ \sum_{j=1}^{m} w_j(p)\,D_{\mathrm{KL}}\!\bigl(p^{(j)}\|\pi_j\bigr)
    \ +\ \sum_{j=1}^{m} w_j(p)\,\log\!\frac{w_j(p)}{w_j}.
  \label{eq:KL-block-decomp}
\end{equation}
\end{lemma}

\begin{proof}
Using the disjointness $\mathcal{X}=\bigsqcup_j \mathcal{X}_j$ and writing $p(x)=w_j(p)\,p^{(j)}(x)$ for $x\in\mathcal{X}_j$,
\begin{align}
D_{\mathrm{KL}}(p\|\pi)
&= \sum_{j=1}^{m} \sum_{x\in\mathcal{X}_j} p(x)\,\log\!\frac{p(x)}{w_j\,\pi_j(x)} \notag\\
&= \sum_{j=1}^{m} \sum_{x\in\mathcal{X}_j} w_j(p)\,p^{(j)}(x)
   \left[ \log\!\frac{p^{(j)}(x)}{\pi_j(x)} + \log\!\frac{w_j(p)}{w_j} \right] \notag\\
&= \sum_{j=1}^{m} w_j(p)\,D_{\mathrm{KL}}\!\bigl(p^{(j)}\|\pi_j\bigr)
   + \sum_{j=1}^{m} w_j(p)\,\log\!\frac{w_j(p)}{w_j}.
\end{align}

which is \eqref{eq:KL-block-decomp}.
\end{proof}

\medskip
Equations \eqref{eq:block-mass-preserve}–\eqref{eq:conditional-evolution} show that the dynamics preserves coarse variables $\{w_j\}$ and evolves conditionals within each block, while \eqref{eq:KL-block-decomp} cleanly separates \emph{within-block} divergences from a \emph{coarse-mass mismatch} term. 
These identities will be the backbone for the Lyapunov property established in Sec.~\ref{subsec:lyapunov}.

\subsection{Lyapunov property of the KLD potential}\label{subsec:lyapunov}

\medskip
\noindent\textit{Statement.}
Let $p_{n+1}=Tp_n$ with a block-invariant kernel $T=\mathrm{diag}(T_1,\dots,T_m)$ (Sec.~\ref{subsec:block-invariance}).
\begin{theorem}[KLD potential is Lyapunov under block invariance]\label{thm:V-Lyapunov}
For all $n\ge0$,
\begin{equation}
V(p_{n+1})\ \le\ V(p_n).
\label{eq:V-monotone}
\end{equation}
If, in addition, each block kernel $T_j$ is primitive (irreducible and aperiodic) with unique invariant $\pi_j^\ast\in\mathcal{I}_j$, then
\begin{equation}
V(p_n)\ =\ \sum_{j=1}^{m} w_j(p_0)\,D_{\mathrm{KL}}\!\bigl(p_n^{(j)}\!\|\pi_j^\ast\bigr)
\ \xrightarrow[n\to\infty]{}\ 0.
\label{eq:V-vanish}
\end{equation}
\end{theorem}

\begin{proof}
Fix $\pi\in\Pi(p_0)$ and write $\pi=\sum_{j=1}^{m} w_j(p_0)\,\pi_j$ with $\pi_j\in\mathcal{I}_j$. 
Block masses are conserved along the trajectory (Lemma~\ref{lem:block-preservation}), hence $w_j(Tp)=w_j(p)=w_j(p_0)$. 
By the KL block decomposition (Lemma~\ref{lem:block-decomposition}),
\begin{equation}
\begin{aligned}
D_{\mathrm{KL}}(Tp\,\|\,\pi)
&= \sum_{j=1}^{m} w_j(p_0)\,D_{\mathrm{KL}}\!\bigl((Tp)^{(j)}\,\|\,\pi_j\bigr)\\
&= \sum_{j=1}^{m} w_j(p_0)\,D_{\mathrm{KL}}\!\bigl(T_j p^{(j)}\,\|\,\pi_j\bigr).
\end{aligned}
\label{eq:KL-after-step}
\end{equation}
Applying the data-processing inequality (DPI) blockwise,
\begin{equation}
D_{\mathrm{KL}}(T_j r\,\|\,\pi_j)\ \le\ D_{\mathrm{KL}}(r\,\|\,\pi_j)\qquad(\pi_jT_j=\pi_j),
\label{eq:KL-contract}
\end{equation}
and inserting \eqref{eq:KL-contract} into \eqref{eq:KL-after-step} yields
\begin{equation}
D_{\mathrm{KL}}(Tp\,\|\,\pi)\ \le\ \sum_{j=1}^{m} w_j(p_0)\,D_{\mathrm{KL}}\!\bigl(p^{(j)}\,\|\,\pi_j\bigr)
= D_{\mathrm{KL}}(p\,\|\,\pi).
\label{eq:one-step-KL}
\end{equation}
Taking the infimum over $\pi\in\Pi(p_0)$ proves \eqref{eq:V-monotone}. 
If each $T_j$ is primitive with invariant $\pi_j^\ast$, then $p_n^{(j)}\to\pi_j^\ast$ and \eqref{eq:V-vanish} follows.
\end{proof}

\section*{Robustness to weak inter-block leakage and re-partitioning}
\label{sec:robustness}

\paragraph*{Motivation.}
Secs.~\ref{subsec:lemmas}--\ref{subsec:lyapunov} establish a Lyapunov principle for $V$ under
\emph{exact} block invariance. To examine robustness to \emph{weak} violations (small inter-block
leakage or re-partitioning), we introduce a leakage-tolerant potential $V_\delta$ that continuously
extends $V$ and preserves its qualitative behavior up to $O(\delta)$ slack.

\paragraph{Leakage-tolerant admissible set and closed form (corrected).}
When leakage is small over the observation window, relax coarse masses by $\pm\delta$ and define
\[
\Pi_{\mathcal{P},\delta}(p_0)
:=\Bigl\{\,
\pi\in\Delta(\mathcal{X})\ \Big|\ 
\forall j,\ \Bigl|\sum_{i\in\mathcal{X}_j}\pi_i - w_j(p_0)\Bigr|\le \delta
\,\Bigr\},\qquad
V_\delta(p):=\inf_{\pi\in\Pi_{\mathcal{P},\delta}(p_0)} D_{\mathrm{KL}}(p\Vert \pi).
\]
By KL block decomposition and KKT optimality, the optimizer has the form
\[
w_j^\star \;=\; \biggl[\,\frac{w_j(p)}{\tau^\star}\,\biggr]_{[\,a_j,\,b_j\,]},
\quad
a_j:=w_j(p_0)-\delta,\ \ b_j:=w_j(p_0)+\delta,
\]
where $\tau^\star>0$ is chosen so that $\sum_j w_j^\star=1$. Hence
\begin{equation}
\boxed{\;
V_\delta(p)
=\sum_{j=1}^{m} w_j(p)\,
\log\!\frac{w_j(p)}{\,w_j^\star\,},
\quad
w_j^\star=\biggl[\,\frac{w_j(p)}{\tau^\star}\,\biggr]_{[\,a_j,\,b_j\,]},
\ \ \sum_j w_j^\star=1.
\;}
\label{eq:Vdelta-correct}
\end{equation}
Since $f(\tau):=\sum_j \bigl[w_j(p)/\tau\bigr]_{[a_j,b_j]}$ is continuous and strictly decreasing in $\tau>0$, $\tau^\star$ is found by a short bisection on $(0,\infty)$.

A self-contained derivation and a sufficient threshold for strict decrease under leakage are provided in Appendix A.

\paragraph{Continuity and small-leakage persistence.}
As $\delta\to 0$, $V_\delta(p)\to V(p)$. Hence, under sufficiently small leakage, the qualitative
properties proved for $V$ (monotonicity along the dynamics, long-time limit, and the thermodynamic
reading $k_{\mathrm B}T\,V_\delta$ per refresh) persist up to an $O(\delta)$ relaxation.

\paragraph{Coarsening by merging leaking blocks.}
If mutually leaking blocks are merged to form a coarser partition $\mathcal{Q}\succeq\mathcal{P}$,
data processing implies $V_{\mathcal{Q}}(p)\le V_{\mathcal{P}}(p)$. Coarsening restores exact
block invariance at the partition level, at the price of a more conservative (smaller) KL-based
potential and a weaker maintenance-cost bound.

\section{Instantiations}\label{sec:instantiations}

\subsection{Image replication via Gaussian convolution}\label{subsec:image}

\paragraph{Block-patterned model.}
First, we model an image with an explicit \emph{block (mosaic) structure}: the pixel grid is partitioned
into $m$ disjoint rectangular regions that remain isolated in the non-ergodic variant (no
cross-region mass transfer). Figure~\ref{fig:image} uses three blockings of the same $256\times256$
image: $2\times2$ (top), $4\times4$ (middle), and $128\times128$ (bottom). The ergodic variant
corresponds to the same update without any blocking (global mixing).

\medskip
We model a single replication step as a Markov smoothing map $T_\sigma$ \emph{implemented} by a
discrete Gaussian convolution on $I:\mathbb{Z}^2\to[0,1]$ followed by re-quantization. The kernel \emph{coincides with} $T_\sigma$ here. The kernel (discrete, renormalized) is given by:
\begin{equation}
  T_\sigma(u,v)\;:=\;
  \frac{\exp\!\bigl(-(u^2+v^2)/(2\sigma^2)\bigr)}
       {\displaystyle\sum_{a,b\in\mathbb{Z}}\exp\!\bigl(-(a^2+b^2)/(2\sigma^2)\bigr)},
  \qquad (u,v)\in\mathbb{Z}^2,\ \sigma>0,
  \label{eq:gauss-kernel}
\end{equation}
so that $\sum_{u,v}T_\sigma(u,v)=1$. The (global) discrete convolution is
\begin{equation}
  (I\ast T_\sigma)(x,y)\ :=\ \sum_{u,v\in\mathbb{Z}} I(x-u,y-v)\,T_\sigma(u,v),
  \label{eq:image-conv}
\end{equation}
and in the non-ergodic (blockwise) variant the sum is restricted to pixels belonging to the same
predefined region, which simply makes the $T_\sigma$ block diagonal.

Let $p_n$ denote the (normalized) mosaic block distribution, histogram,  before step $n$. The convolution \eqref{eq:image-conv} followed by the same binning induces the one-step Markov channel in the mosaic block distribution, denoted $T_\sigma$ so that
\begin{equation}
  q_n\ =\ p_n\,T_\sigma,\qquad p_{n+1}\ =\ q_n,
  \label{eq:image-update}
\end{equation}
with row vectors acting on the right.

\paragraph{Recorded metrics.}
We view the histogram on a finite state space $\mathcal X$ (intensity bins). 
\begin{align}
  H(q_n)        &= -\sum_{x\in\mathcal X} q_n(x)\,\log q_n(x), \label{eq:H-image}\\
  H_\times(n)   &= -\sum_{x\in\mathcal X} p_n(x)\,\log q_n(x), \label{eq:Hcross-image}\\
  D_{\mathrm{KL}}(p_n\|q_n)
                &= \sum_{x\in\mathcal X} p_n(x)\,\log\frac{p_n(x)}{q_n(x)}, \label{eq:KL-image}\\
  V(p_n)        &= \inf_{\pi\in\Pi(p_0)} D_{\mathrm{KL}}(p_n\|\pi). \label{eq:V-image}
\end{align}

For Fig.~\ref{fig:image} (Appendix B), values are shown in \emph{bits} (base 2).
Each panel annotates a certified $\varepsilon$ and the corresponding $\delta(T_\sigma)$ for its
(blockwise) kernel.

\paragraph*{Changing minimizers.}
When a block admits multiple invariant measures ($\mathcal I_j$ not a singleton, e.g.\ reducible blocks),
the minimizer $\pi^\star\!\in\!\Pi(p_0)$ of $V(p_n)$ can \emph{change along the trajectory} even though
$V$ remains nonincreasing. A concrete three-state reducible example with an explicit switching of the KLD minimizer is given in Appendix G.

\FloatBarrier
\begin{figure*}[p]
  \centering
  \hspace*{\figindent}%
  \begin{minipage}{\dimexpr\textwidth-\figindent\relax}
    \centering
    \includegraphics[width=0.70\linewidth]{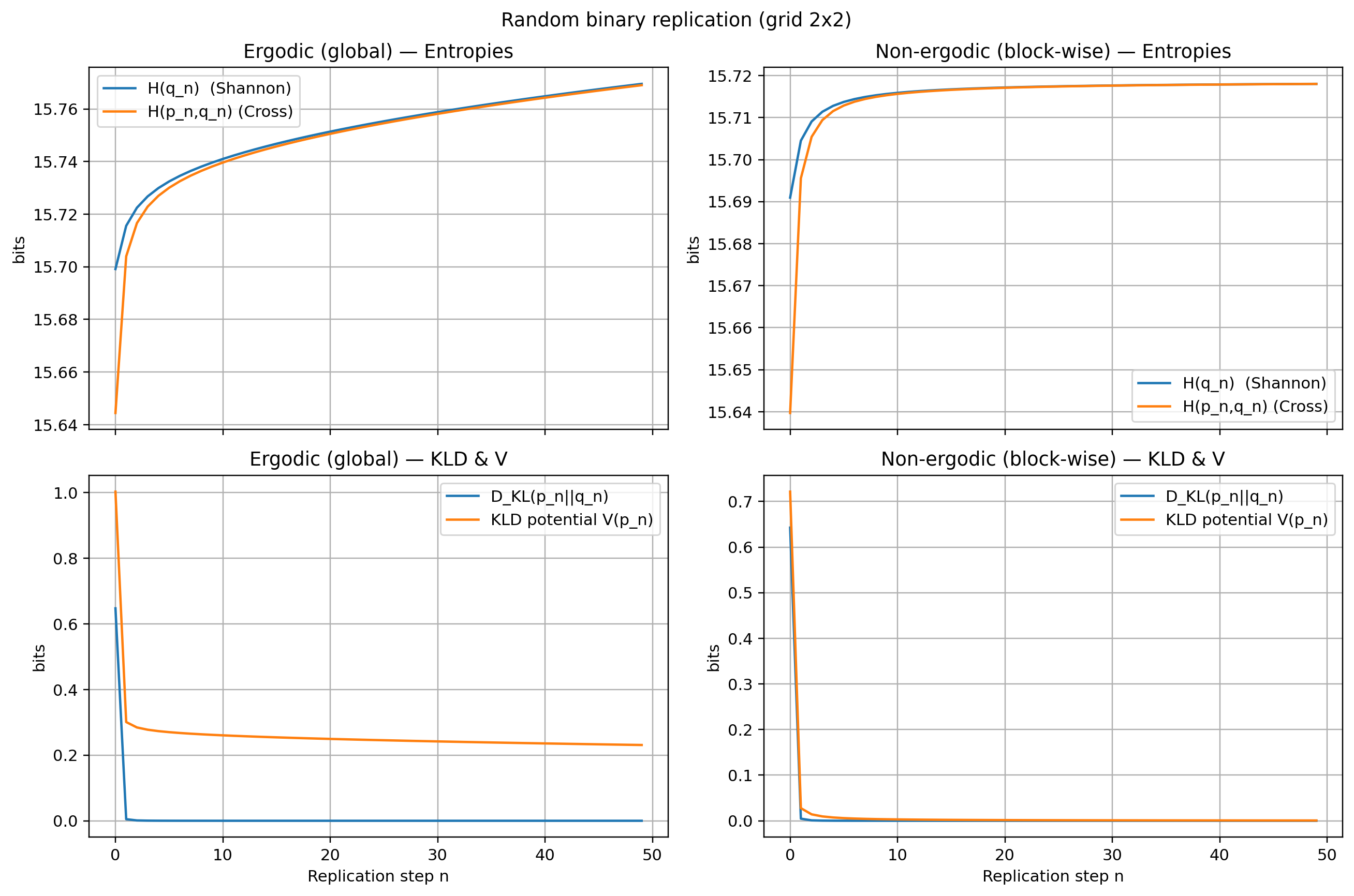}\vspace{0.5em}
    \includegraphics[width=0.70\linewidth]{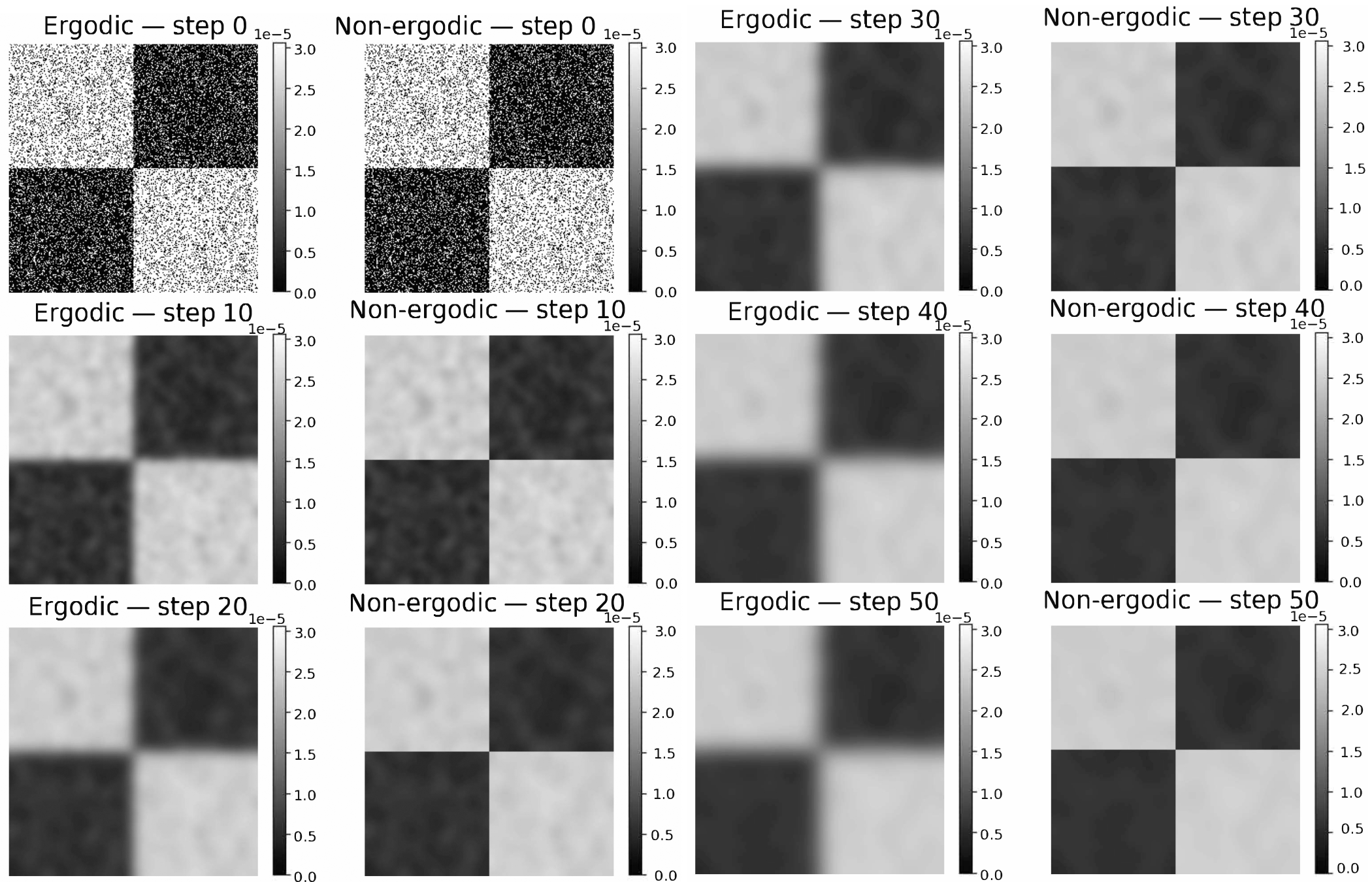}
    \blockstats{\epsTwo}{\delTwo} 
  \end{minipage}
\end{figure*}

\begin{figure*}[p]\caption{}\ContinuedFloat
  \centering
  \hspace*{\figindent}%
  \begin{minipage}{\dimexpr\textwidth-\figindent\relax}
    \centering
    \includegraphics[width=0.70\linewidth]{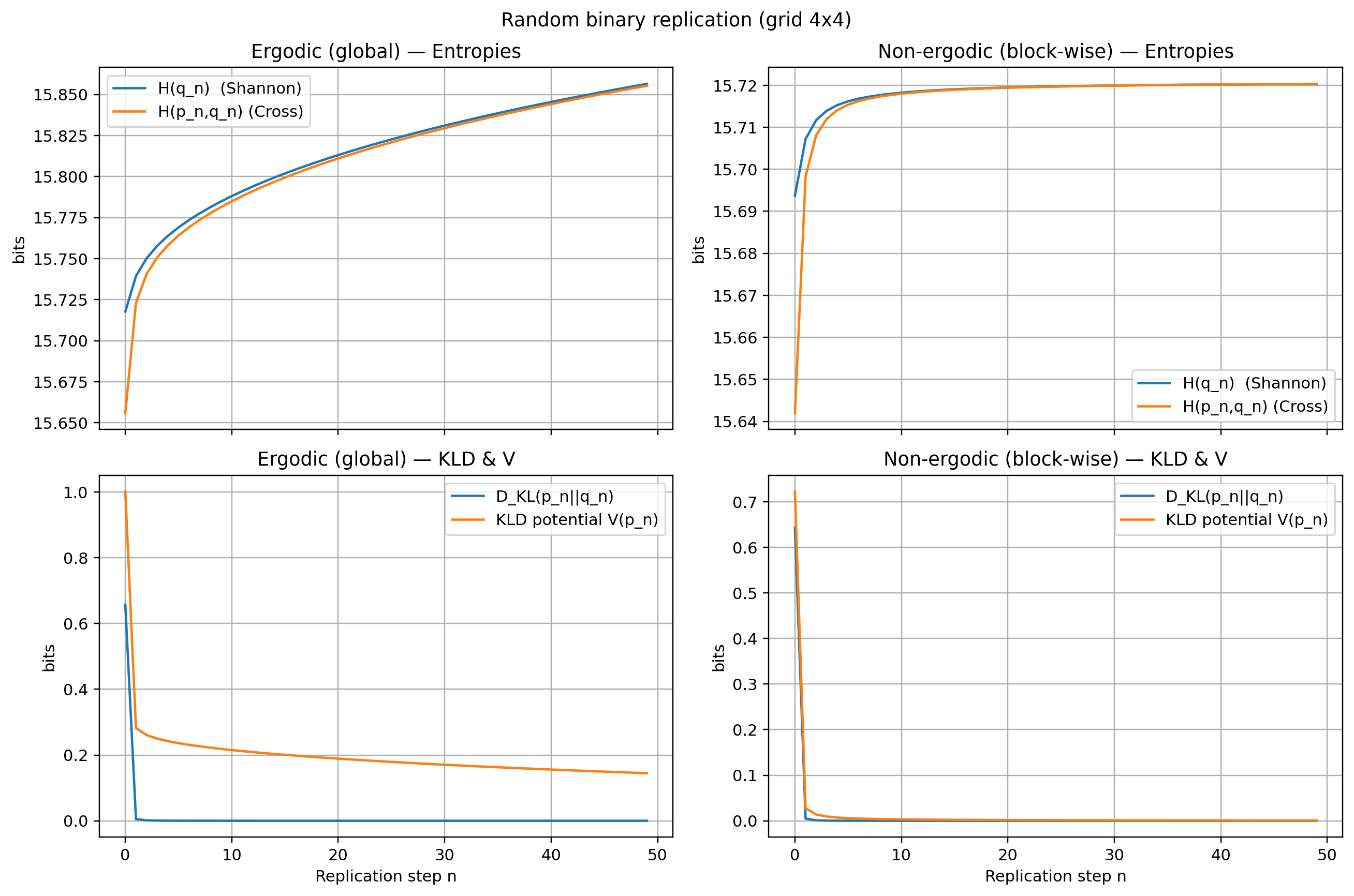}\vspace{0.5em}
    \includegraphics[width=0.70\linewidth]{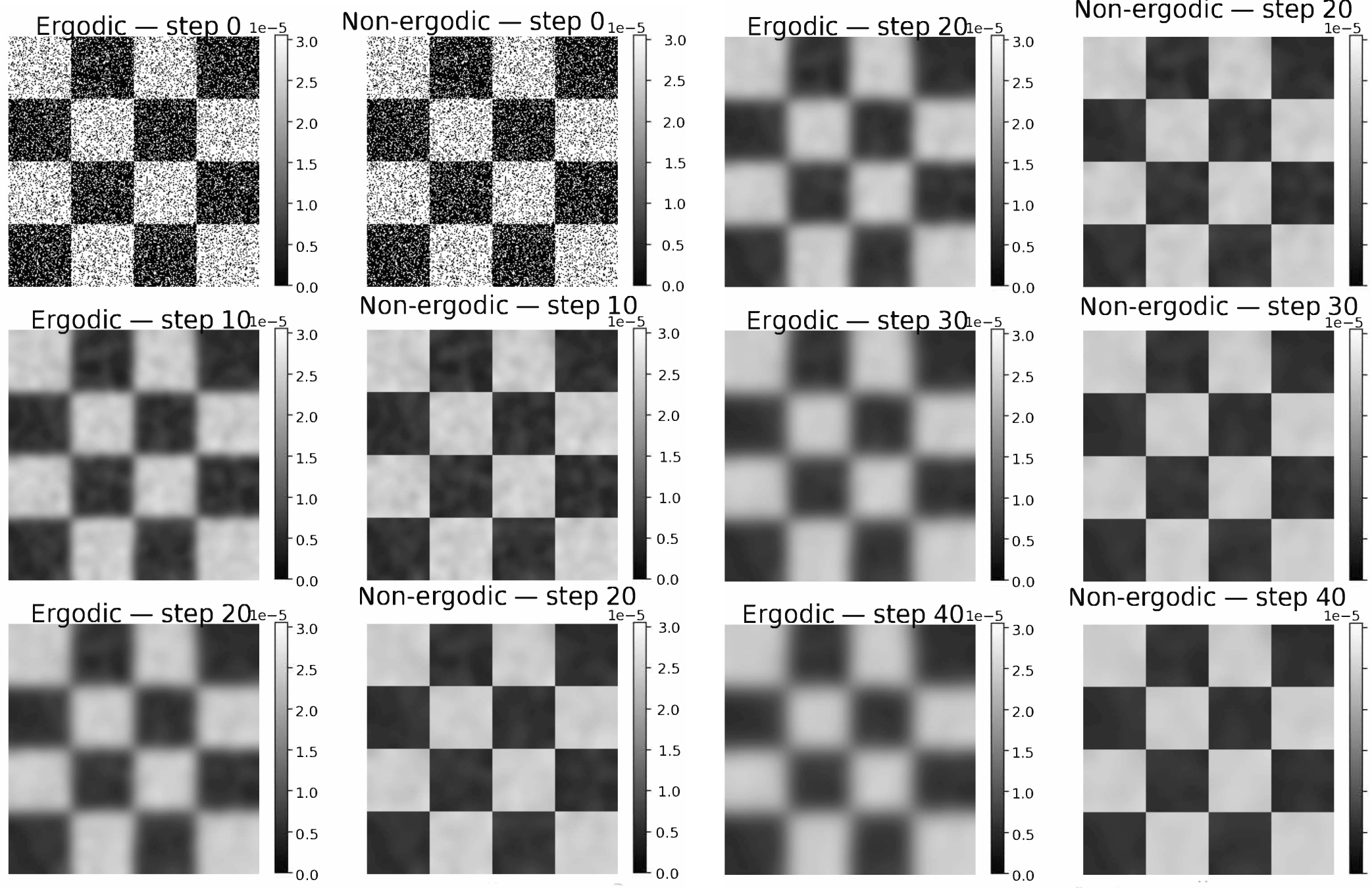}
    \blockstats{\epsFour}{\delFour} 
  \end{minipage}
\end{figure*}

\begin{figure*}[p]\ContinuedFloat
  \centering
  \hspace*{\figindent}%
  \begin{minipage}{\dimexpr\textwidth-\figindent\relax}
    \centering
    \includegraphics[width=0.70\linewidth]{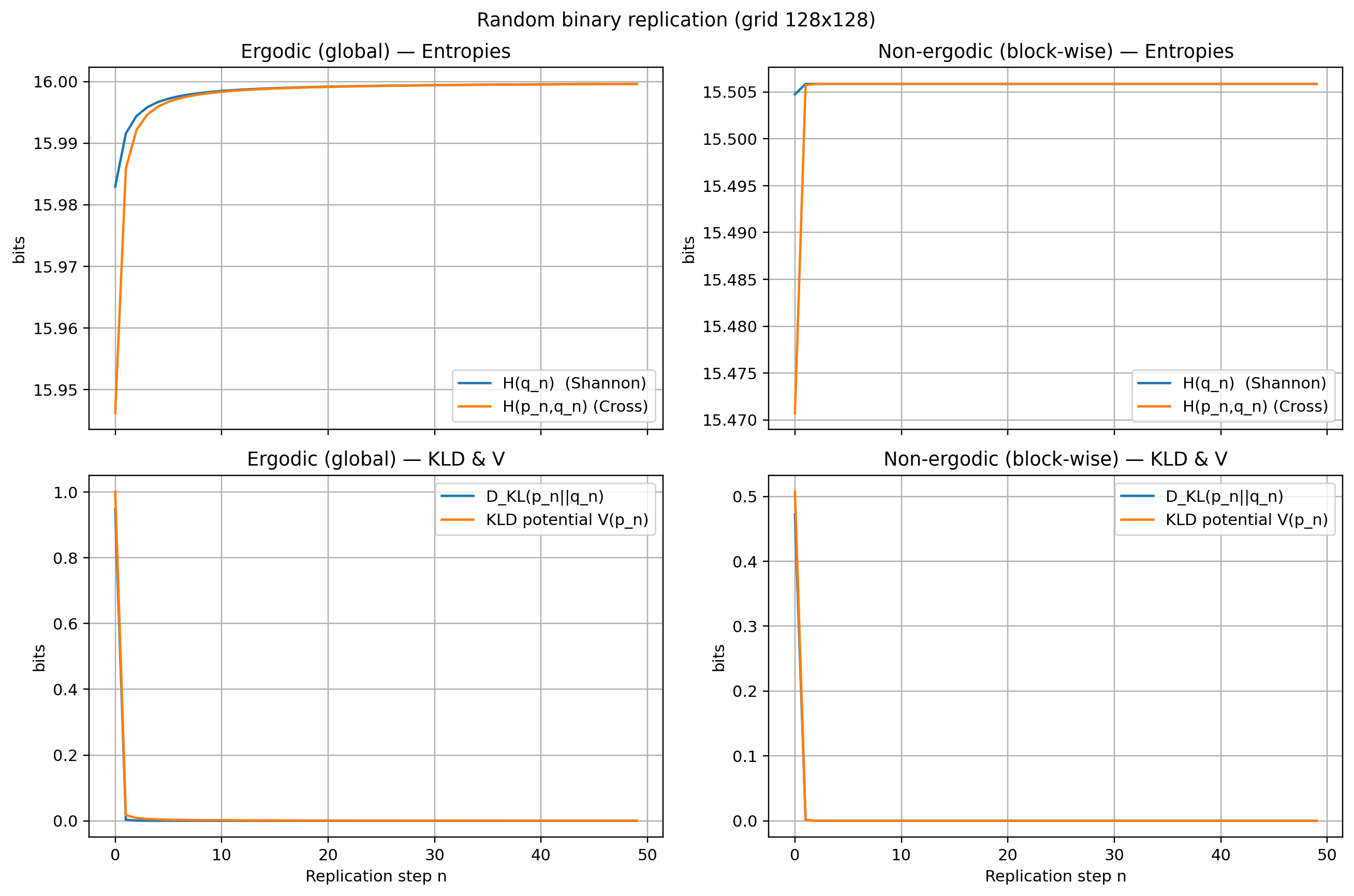}\vspace{0.5em}
    \includegraphics[width=0.70\linewidth]{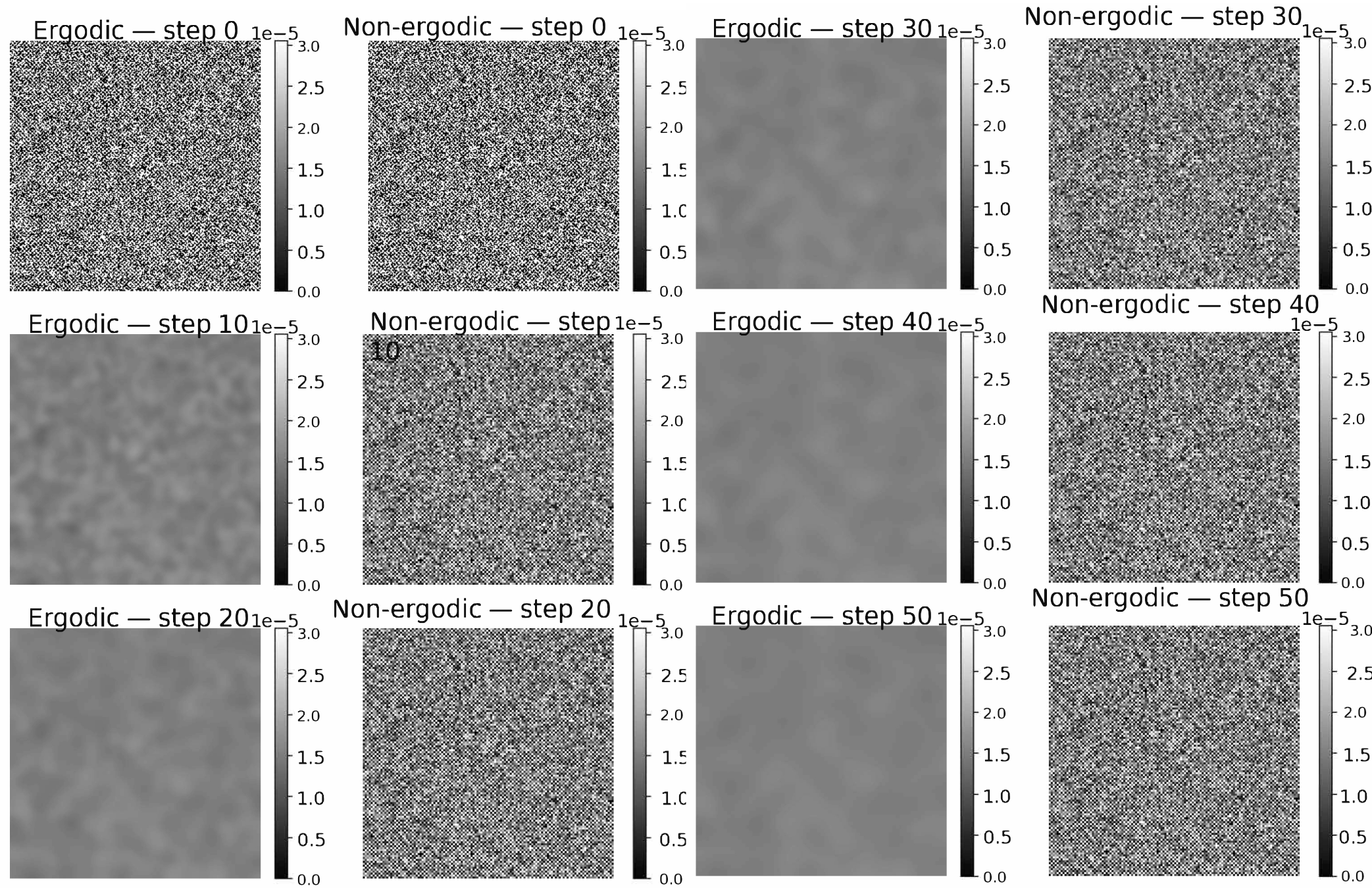}
    \blockstats{\epsHuge}{\delHuge} %
  \end{minipage}
 
  \caption[Ergodic vs.\ non-ergodic Gaussian blurring (bits)]{Ergodic (global) vs.\ non-ergodic (blockwise) Gaussian blurring of a $256\times256$ image ($\sigma=1.5$, $n_{\mathrm{steps}}=50$).
  At each step we report $H$~\eqref{eq:H-image}, $H_\times$~\eqref{eq:Hcross-image}, $D_{\mathrm{KL}}$~\eqref{eq:KL-image}, and $V$~\eqref{eq:V-image}.
  In the blockwise case, monotonicity of $V$ follows from Theorem~\ref{thm:V-Lyapunov}.
  (Upper) $2\times2$ blocks, (Middle) $4\times4$ blocks, (Lower) $128\times128$ blocks.
  All values in this figure are in \emph{bits} (base~2; see Appendix B).
  \par\smallskip
  \textit{Per-panel mixing strength (see Defs.~\eqref{eq:Doeblin-cond}--\eqref{eq:Dobrushin-def}).}
  Upper:\ $\varepsilon=\epsTwo$, $\delta=\delTwo$;\quad
  Middle:\ $\varepsilon=\epsFour$, $\delta=\delFour$;\quad
  Lower:\ $\varepsilon=\epsHuge$, $\delta=\delHuge$.}
  \label{fig:image}
\end{figure*}

\subsection{DNA replication as a block-diagonal substitution process}\label{subsec:dna}

\paragraph{Biophysical context.}
Second, we model DNA replication with the proofreading/repairing mechanism after copying. DNA polymerases move on copied DNA stochastically with forward/backward steps, proofreading, and context-dependent kinetics \cite{Wuite2000,Bustamante2003,Le2018,Gaspard2016a,Pineros2020}. 
To capture non-ergodicity induced by compositional structure (e.g. AT- vs. GC-rich segments), we model replication as a \emph{block-diagonal} Markov process on nucleotides.

\paragraph{Modeling stance.}
The AT/GC split is a biologically motivated \emph{stylized} coarse-graining:
it captures slow exchange between compositionally distinct pools. Small AT$\leftrightarrow$GC
exchange can be handled by the leakage tolerance potential $V_\delta$ or by
coarsening the partition (see ``Robustness to weak interblock leakage'').

\paragraph{Block-diagonal dynamics.}
Let $\mathcal{X}=\{\mathrm{A},\mathrm{T},\mathrm{C},\mathrm{G}\}$ with two blocks
$\mathcal{X}_1=\{\mathrm{A},\mathrm{T}\}$ and $\mathcal{X}_2=\{\mathrm{C},\mathrm{G}\}$. 
A one-step substitution kernel (copying) with block invariance \eqref{eq:block-invariance} is
\begin{equation}
\begin{aligned}
T   &= \operatorname{diag}(T_1,T_2),\\
T_1 &= \begin{pmatrix}
        1-\alpha & \alpha\\
        \beta    & 1-\beta
      \end{pmatrix},
\qquad
T_2 = \begin{pmatrix}
        1-\gamma & \gamma\\
        \delta   & 1-\delta
      \end{pmatrix}.
\end{aligned}
\label{eq:T-block}
\end{equation}

with $\alpha,\beta,\gamma,\delta\in(0,1)$ encoding effective substitution tendencies per step (arising from selectivity, context, and post-replicative processing). 
The block masses $w_1=p_{\mathrm A}+p_{\mathrm T}$ and $w_2=p_{\mathrm C}+p_{\mathrm G}$ are conserved by the lemma~\ref{lem:block-preservation}.

\paragraph{Proofreading/repair channel.}
To incorporate proofreading/repair, we add a blockwise operator $\mathcal{R}=\mathrm{diag}(R_1,R_2)$ that is invoked with probability $\rho\in[0,1]$ after extension:
\begin{equation}
  R_1=\begin{pmatrix}1-\alpha' & \alpha'\\[2pt]\beta' & 1-\beta'\end{pmatrix},\qquad
  R_2=\begin{pmatrix}1-\gamma' & \gamma'\\[2pt]\delta' & 1-\delta'\end{pmatrix},
  \label{eq:R-block}
\end{equation}
with typically $\alpha',\beta',\gamma',\delta' < \alpha,\beta,\gamma,\delta$ (improved selectivity).
The effective one-step kernel is the convex mixture.
\begin{equation}
  \widetilde{T}\ :=\ (1-\rho)\,T\ +\ \rho\,\mathcal{R},
  \label{eq:T-tilde}
\end{equation}
which remains block diagonal and therefore satisfies \eqref{eq:block-invariance}. 
The replication update reads
\begin{equation}
  q_n\ =\ p_n\,\widetilde{T},\qquad p_{n+1}\ =\ q_n.
  \label{eq:dna-update}
\end{equation}

\noindent\textit{Primitivity in this setting.}
For a two-state block $[[1-a,a],[b,1-b]]$, primitiveness is equivalent to $a>0$ and $b>0$ with $(a,b)\neq(1,1)$; in particular, $a,b\in(0,1)$ suffices. Under the mixture $\widetilde T=(1-\rho)T+\rho\mathcal R$, if $\mathcal R$ has strictly positive entries then any $\rho>0$ makes all effective entries positive, hence each block becomes primitive and $V(p_n)\to 0$.

\paragraph{KLD potential with unique block invariants.}
Assume that each block chain admits a unique invariant $\pi_j^\ast\in\mathcal{I}_j$. Then, combining \eqref{eq:V-def-intro} with \eqref{eq:KL-block-decomp} gives
\begin{equation}
  V(p_n)\ =\ \sum_{j=1}^{2} w_j(p_0)\,
  D_{\mathrm{KL}}\!\bigl(p_n^{(j)}\ \|\ \pi_j^\ast\bigr).
  \label{eq:V-dna-explicit}
\end{equation}
By Theorem~\ref{thm:V-Lyapunov}, $V$ is a Lyapunov function for the blockwise dynamics:
$V(p_{n+1})\le V(p_n)$ for all $n$, and $V(p_n)\to 0$ as $n\to\infty$.

\paragraph{Local detailed balance derivation.}
Within each two-state block, write the undriven rates as
\[
A \xrightleftharpoons[\ \beta\ ]{\ \alpha\ } T
\qquad\text{and}\qquad
C \xrightleftharpoons[\ \delta\ ]{\ \gamma\ } G.
\]
At equilibrium (no chemical drive), detailed balance with the block's stationary distribution $\pi^\ast$ gives
\[
\pi^{\ast}_{\!A}\,\alpha=\pi^{\ast}_{\!T}\,\beta
\ \Rightarrow\
\frac{\alpha}{\beta}
=\frac{\pi^{\ast}_{\!T}}{\pi^{\ast}_{\!A}},
\qquad
\frac{\gamma}{\delta}
=\frac{\pi^{\ast}_{\!G}}{\pi^{\ast}_{\!C}}.
\]
When a proofreading cycle provides a chemical affinity $\Delta\mu>0$ per round-trip, the local detailed balance (cycle thermodynamics) shifts \emph{log rate odds} by $\Delta\mu/k_{\mathrm B}T$:
\begin{equation}
\ln\frac{\alpha'}{\beta'}
=\ln\frac{\alpha}{\beta}+\frac{\Delta\mu}{k_{\mathrm B}T},
\qquad
\ln\frac{\gamma'}{\delta'}
=\ln\frac{\gamma}{\delta}+\frac{\Delta\mu}{k_{\mathrm B}T}.
\label{eq:LDB-shift}
\end{equation}
Equivalently, in odds form,
\begin{equation}
\frac{\alpha'}{\beta'} \;=\frac{\alpha}{\beta}\,\exp\!\Big(\frac{\Delta\mu}{k_{\mathrm B}T}\Big),
\qquad
\frac{\gamma'}{\delta'} \;=\frac{\gamma}{\delta}\,\exp\!\Big(\frac{\Delta\mu}{k_{\mathrm B}T}\Big).
\label{eq:proofreading-affinity-rewrite}
\end{equation}
Thus, the drive fixes \emph{odds} (ratios) but not individual rates; pinning $\alpha',\beta'$ (or $\gamma',\delta'$) separately requires an additional timescale constraint (for example, keeping $\alpha'+\beta'$ fixed), which we do not assume here. As is standard in stochastic thermodynamics, the exponential tilt of rate odds by the chemical cycle affinity follows from network cycle thermodynamics. Equations~\eqref{eq:LDB-shift}–\eqref{eq:proofreading-affinity-rewrite} encode local detailed balance for \emph{rate odds}; they shape steady odds and relaxation speed but do \emph{not} by themselves imply primitivity, which is a structural (irreducible+aperiodic) property; cf.\ Theorem~\ref{thm:ctmc-primitive}.

\noindent\emph{Note.} The exponential tilt in Eq.~\eqref{eq:proofreading-affinity-rewrite} biases \emph{odds} but does not by itself specify the steady entropy production rate; the latter can remain positive even when $V\to0$ under primitive blocks.

\paragraph{Energetic bias (phenomenology).}
In kinetic proofreading a chemical drive provides an affinity $\Delta\mu$ per cycle (e.g., via nucleotide hydrolysis). Within each invariant block, this drive tilts the two-state odds toward the proofreading-favored direction according to Eq.~\eqref{eq:proofreading-affinity-rewrite}, sharpening the steady composition. Because $V$ decomposes over blocks, a stronger drive typically steepens the within-block basins (reducing the Bernoulli variance around $\pi_j^\ast$) and thus accelerates the monotone decay of $V$ guaranteed by Theorem~\ref{thm:V-Lyapunov}. Crucially, the existence and monotonicity of $V$ do \emph{not} require a drive; the drive only modifies the speed and the steady odds.

\paragraph{Primitivity of $2\times2$ blocks and relation to Eq.~(35).}
For the two-state blocks in Example~2, a practical sufficient condition for block primitivity is strict positivity of every entry of the effective $2\times2$ transition kernel:
$\alpha_{\rm eff},\beta_{\rm eff},\gamma_{\rm eff},\delta_{\rm eff}\in(0,1)$. Then each block is irreducible and aperiodic (hence primitive), and \eqref{eq:V-dna-explicit} implies $V(p_n)\to0$.
If $\alpha_{\rm eff}=0$ or $\beta_{\rm eff}=0$, irreducibility fails; if $\alpha_{\rm eff}=\beta_{\rm eff}=1$ (and analogously for CG), period-2 oscillations arise unless additional repair/proofreading steps restore positive diagonal entries. 
Equation~\eqref{eq:proofreading-affinity-rewrite} (Eq.~35) expresses how $\Delta\mu$ biases \emph{rate ratios}; it determines the steady odds and relaxation speed but is not, by itself, a primitivity condition, which depends on positivity and aperiodicity of the effective kernels.

\smallskip
\noindent\emph{Notes.}
 The sign of $\Delta\mu$ sets the favored direction; reversing the operating bias corresponds to $\Delta\mu\!\to\!-\Delta\mu$.  At stationarity the block-wise odds inherit the exponential tilt:
$\pi_A/\pi_T=\beta'/\alpha'$ and $\pi_C/\pi_G=\delta'/\gamma'$, consistent with the proofreading picture of Hopfield~\cite{Hopfield1974}. This connects proofreading accuracy to energetic cost in line with kinetic proofreading thermodynamics \cite{Hopfield1974,SartoriPigolotti2015}.

\paragraph{Recorded metrics and bookkeeping.}

\noindent At each step, we record:  \(H(q_n)\), \(H_\times(n)=H(p_n,q_n)\), \(D_{\mathrm{KL}}(p_n\|q_n)\), and \(V(p_n)\),
with monotone $V$ ensured by Theorem~\ref{thm:V-Lyapunov}. 
Because $\widetilde T$ is block diagonal by construction, coarse transitions between $\{\mathrm{AT},\mathrm{GC}\}$ vanish at one step. Block-level confusion rates under small inter-block leakage are analyzed in the leakage section.

\paragraph{Readouts and figures.}
The trajectories in Fig.2 show an increase in $H$ and $H_\times$, a decrease in $D_{\mathrm{KL}}(p_n\|q_n)$, and a monotone decay of $V(p_n)$, in agreement with Theorem~\ref{thm:V-Lyapunov}. 
Figure 3 shows the potential landscape $V(p)=w_1 V_1(x)+w_2 V_2(y)$ on $(x,y)$, where $x$ is the A fraction in $\{\mathrm A,\mathrm T\}$ and $y$ the C fraction in $\{\mathrm C,\mathrm G\}$(Appendix D).

\begin{figure}[t]
  \centering
  \begin{minipage}{0.8\textwidth}
    \centering
    \includegraphics[width=\linewidth]{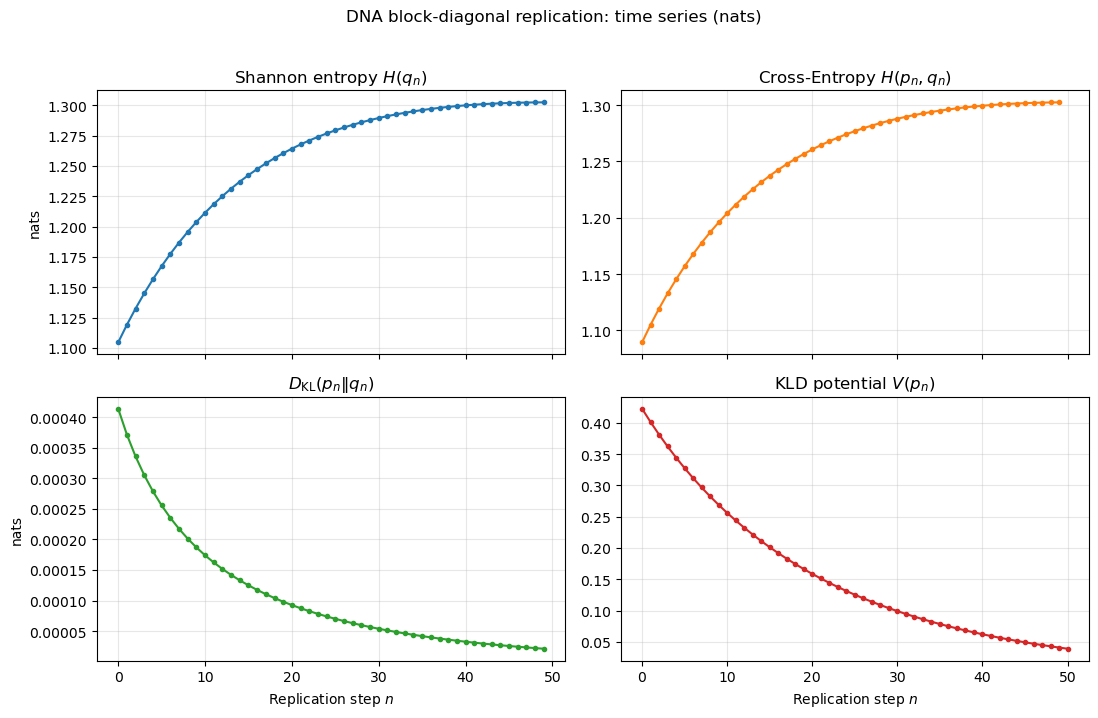}
    \vspace{1pt}\par\small 
  \end{minipage}\hfill
\caption{DNA block-diagonal replication: time series of informational metrics (nats).Shannon entropy, Cross entropy, KLD, and KLD potential.
We iterate a four-state Markov chain on $\{\mathrm{A},\mathrm{T},\mathrm{C},\mathrm{G}\}$ with a block-diagonal kernel $\widetilde T=\mathrm{diag}(T_1,T_2)$ for $N=50$ steps.
Base rates: $\alpha=0.020$, $\beta=0.010$, $\gamma=0.015$, $\delta=0.015$.
Proofreading channel $\mathcal R$ with $\alpha'=0.005$, $\beta'=0.003$, $\gamma'=\delta'=0.004$ is mixed at $\rho=0.30$, giving effective rates
$\alpha_{\!e}=0.0155$, $\beta_{\!e}=0.0079$, $\gamma_{\!e}=\delta_{\!e}=0.0117$ and block invariants
$\pi_1^\ast=(0.3376,\,0.6624)$, $\pi_2^\ast=(0.5,\,0.5)$.
Initial distribution $p_0=(0.6,0.1,0.2,0.1)$ so that $w_1(p_0)=0.7$, $w_2(p_0)=0.3$.
At each step $q_n=p_n\widetilde T$ and we record: top-left, Shannon entropy $H(q_n)$; top-right, cross-entropy $H(p_n,q_n)$; bottom-left, step divergence $D_{\mathrm{KL}}(p_n\|q_n)$; bottom-right, KLD potential
$V(p_n)=w_1(p_0)\,D_{\mathrm{KL}}(p_n^{(1)}\|\pi_1^\ast)+w_2(p_0)\,D_{\mathrm{KL}}(p_n^{(2)}\|\pi_2^\ast)$.
As replication proceeds, $H$ and $H_\times$ increase, while both $D_{\mathrm{KL}}(p_n\|q_n)$ and $V(p_n)$ decrease monotonically; the monotonicity of $V$ follows from Theorem~\ref{thm:V-Lyapunov}.All metrics are reported in \emph{nats}}

\end{figure}
\FloatBarrier
The contour plot shows the KLD potential \(V(p)=w_1 D_{\mathrm{KL}}\!\bigl(p^{(1)}\!\parallel\!\pi_1^\ast\bigr)+w_2 D_{\mathrm{KL}}\!\bigl(p^{(2)}\!\parallel\!\pi_2^\ast\bigr)\) on \((x,y)\in[0,1]^2\), where \(x\) is the fraction A in the AT block and \(y\) is the fraction C in the GC block. The unique minimum (blue marker) occurs at \((x^\ast,y^\ast)=(\pi_1^\ast(\mathrm A),\pi_2^\ast(\mathrm C))\approx(0.338,0.600)\), i.e., at the block-wise invariants determined by the effective rates. Because \(V\) is a sum of blockwise KL terms, the level sets are nearly axis-aligned ellipses (here close to circular since \(w_1=w_2=0.5\) and the curvatures are similar). A second-order expansion around \((x^\ast,y^\ast)\) gives \(V(x,y)\approx \tfrac{w_1}{2\,\mu(1-\mu)}(x-\mu)^2+\tfrac{w_2}{2\,\nu(1-\nu)}(y-\nu)^2\) with \(\mu=\pi_1^\ast(\mathrm A)\approx0.338\) and \(\nu=\pi_2^\ast(\mathrm C)=0.600\), yielding coefficients \(\simeq1.12\) and \(\simeq1.04\) (nats), respectively—consistent with the nearly isotropic contours. Increasing proofreading bias decreases \(\mu(1-\mu)\) or \(\nu(1-\nu)\), steepening the basin and accelerating the monotonic decay of \(V(p_n)\) (Theorem~\ref{thm:V-Lyapunov}). The separable, convex landscape highlights that, under block invariance, the non-ergodic relaxation proceeds independently within blocks and converges to the reachable steady point \((x^\ast,y^\ast)\)(Fig. 3).

\begin{figure}[t]
  \centering
  \includegraphics[width=0.5\textwidth]{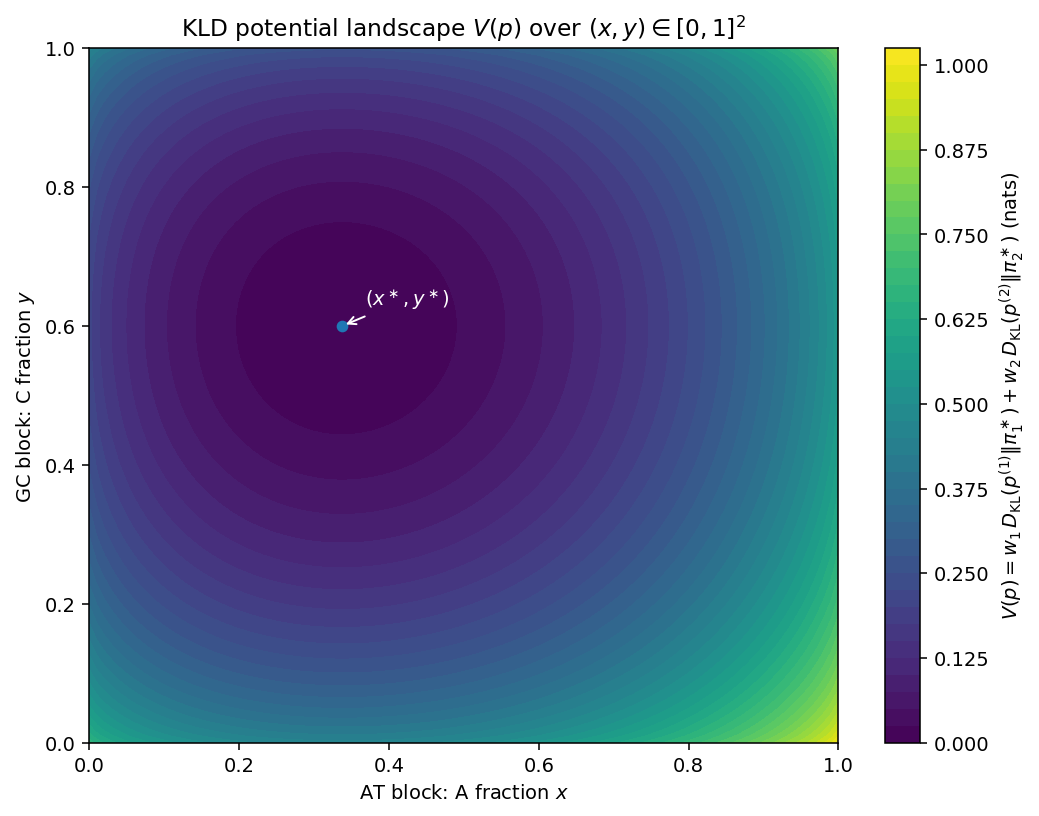}
 \caption{KLD potential landscape $V(p)=w_1 D_{\mathrm{KL}}(p^{(1)}\parallel\pi_1^\ast)+w_2 D_{\mathrm{KL}}(p^{(2)}\parallel\pi_2^\ast)$ (nats) over $(x,y)\in[0,1]^2$, where $x$ is the A fraction in the AT block and $y$ is the C fraction in the GC block. Proofreading mixture $\rho=0.30$; AT rates $\alpha=0.020$, $\beta=0.010$, $\alpha'=0.005$, $\beta'=0.003$ give $\alpha_{\mathrm{eff}}=0.0155$, $\beta_{\mathrm{eff}}=0.0079$ and $\pi_1^\ast=(\beta_{\mathrm{eff}}/(\alpha_{\mathrm{eff}}+\beta_{\mathrm{eff}}),,\alpha_{\mathrm{eff}}/(\alpha_{\mathrm{eff}}+\beta_{\mathrm{eff}}))\approx(0.338,,0.662)$. GC rates (asymmetric) $\gamma=0.014$, $\delta=0.021$, $\gamma'=0.004$, $\delta'=0.006$ give $\gamma_{\mathrm{eff}}=0.0110$, $\delta_{\mathrm{eff}}=0.0165$ and $\pi_2^\ast=(\delta_{\mathrm{eff}}/(\gamma_{\mathrm{eff}}+\delta_{\mathrm{eff}}),,\gamma_{\mathrm{eff}}/(\gamma_{\mathrm{eff}}+\delta_{\mathrm{eff}}))\approx(0.60,,0.40)$. Block weights $w_1=w_2=0.5$. The blue marker indicates $(x^\ast,y^\ast)=(\pi_1^\ast(\mathrm{A}),,\pi_2^\ast(\mathrm{C}))\approx(0.34,,0.60)$. Grid $101\times101$, colormap All values in nats.}
  \label{fig:dna-landscape}
\end{figure}

\FloatBarrier
\noindent\textit{Note on parameter sets.}
Figure 2 uses a symmetric GC block ($\gamma_{\!e}=\delta_{\!e}$; $\pi_2^\ast=(0.5,0.5)$), whereas Figure 3 uses an asymmetric GC block leading to $\pi_2^\ast=(0.60,0.40)$. The two figures illustrate, respectively, symmetric and biased within-block odds.

\section*{An information-theoretic second law}

\paragraph{Standing assumptions.}
We consider a discrete-time Markov evolution $p_{n+1}=T p_n$ on a finite state space,
with a \emph{block-invariant} kernel $T=\mathrm{diag}(T_1,\dots,T_m)$ as in
Sect.~\ref{sec:framework}. Let $q_n:=p_n T$ denote the one-step image of $p_n$.
The KLD potential is
\[
V(p)\;=\;\inf_{\pi\in\Pi(p_0)} D_{\mathrm{KL}}(p\Vert\pi),
\]
where $\Pi(p_0)$ is the reachable steady set (Def.~\eqref{eq:Pi-def}).
By Theorem~\ref{thm:V-Lyapunov}, $V(p_{n+1})\le V(p_n)$ for all $n$.

\paragraph{Step variables and their roles.}
Define the \emph{potential drop} and the \emph{one-step mismatch} by
\begin{equation}
\Delta V_n \;:=\; V(p_n)-V(p_{n+1}),\qquad
D_n \;:=\; D_{\mathrm{KL}}(p_n\Vert q_n),\quad q_n=p_nT.
\label{eq:defs-min-expanded}
\end{equation}
Nonnegativity holds as follows:
(i) $\Delta V_n\ge0$ by the Lyapunov monotonicity of $V$; 
(ii) $D_n\ge0$ because KLD is nonnegative (and $D_n=0$ iff $p_n=q_n$).

\paragraph{Decomposition of the step production.}
Define the \emph{dimensionless step production}
\begin{equation}
\mathcal{S}_n\;:=\;D_n+\Delta V_n.
\label{eq:Sn-def-expanded}
\end{equation}
It gathers the irreversible loss within one update ($D_n$) and the progress toward the reachable steady set ($\Delta V_n$).

\begin{proposition}[Nonnegativity and telescoping sum]
\label{prop:Sn-basic}
For any block-invariant trajectory $\{p_n\}_{n=0}^{N}$,
\begin{equation}
\mathcal{S}_n \;\ge\; 0 \quad (n=0,1,\dots,N-1),
\qquad
\mathcal{S}^{[0:N)}\;:=\;\sum_{n=0}^{N-1}\mathcal{S}_n
\;=\;\sum_{n=0}^{N-1}D_n \;+\; V(p_0)-V(p_N)\;\ge\;0.
\label{eq:sum-Sn-expanded}
\end{equation}
\end{proposition}

\begin{proof}
By \eqref{eq:defs-min-expanded} and Theorem~\ref{thm:V-Lyapunov}, $\Delta V_n\ge0$; KLD nonnegativity gives $D_n\ge0$, hence $\mathcal{S}_n\ge0$. Summing yields $\sum_n\Delta V_n=V(p_0)-V(p_N)$.
\end{proof}

Equation~\eqref{eq:sum-Sn-expanded} states that for any horizon $N$,
\[
\sum_{n=0}^{N-1} D_{\mathrm{KL}}(p_n\Vert p_nT)\;+\;V(p_0)-V(p_N)\;\ge\;0.
\]
Thus the accumulated one-step mismatch and the state-function decrease are jointly nonnegative—an information-theoretic second-law inequality under non-ergodic (block-invariant) dynamics. The key departure from the classical ergodic setting is that $V$ references the \emph{reachable steady set} rather than a single invariant distribution.

\begin{theorem}[Primitivity of block kernels]\label{thm:ctmc-primitive}
Let $T$ be a block transition kernel. If each block is primitive (irreducible and aperiodic), then the system admits a unique invariant distribution within each block, and the dynamics converge to this invariant.\end{theorem}

\begin{theorem}[Second law for non-ergodic replication: minimal form]
\label{thm:second-law-min-expanded}
Equivalently to Proposition~\ref{prop:Sn-basic}, for all $n$ and $N\ge1$,
\begin{equation}
\mathcal{S}_n \;\ge\; 0,
\qquad
\mathcal{S}^{[0:N)}\;\ge\;0.
\end{equation}
In physical units, set $\sigma_{\rm rep}^{(n)} := k_{\mathrm B}T\,\mathcal{S}_n$ and
$\Sigma_{\rm rep}^{(N)} := k_{\mathrm B}T\,\mathcal{S}^{[0:N)}$.
\end{theorem}

\begin{proof}
Immediate from Proposition~\ref{prop:Sn-basic}.
\end{proof}

\begin{remark}[Tightness conditions]\label{rem:strictness-expanded}
$\mathcal{S}_n=0$ iff $D_n=\Delta V_n=0$, i.e., (i) $p_n=q_n$ (a fixed point of $T$), and
(ii) $V(p_{n+1})=V(p_n)$, which requires blockwise tightness of DPI as in Remark~\ref{rem:strictness-expanded}.
Otherwise, $\mathcal{S}_n>0$.
\end{remark}

\paragraph{Thermodynamic reading.}
$V$ measures the distance to the \emph{reachable steady set}; its drop $k_{\mathrm B}T\,\Delta V_n$
plays the role of an informational free-energy decrease, while
$k_{\mathrm B}T\,D_n$ quantifies stepwise dissipative loss of distinguishability.
Hence $\mathcal{S}_n$ acts as a nonnegative production:
\[
\text{(instantaneous dissipation)}\ +\ \text{(state-function decrease)} \;\ge\; 0.
\]

The monotonicity of $V$ uses Lemma~\ref{lem:block-decomposition} and DPI applied \emph{blockwise},
$D_{\mathrm{KL}}(T_j r\Vert \pi_j)\le D_{\mathrm{KL}}(r\Vert \pi_j)$ with $\pi_jT_j=\pi_j$.
If blocks leak, $\Pi(p_0)$ ceases to align with the one-step map and $V$ may increase in one step
We describe the blocks leak in Appendix A.

\paragraph{Extensions (time-dependent/continuous-time).}
(1) \emph{Time-dependent kernels.} If each $T_n$ shares the same block partition holds blockwise, then $V(p_{n+1})\le V(p_n)$ still holds and the theorem remains valid with $q_n=p_n T_n$.
(2) \emph{Continuous-time CTMC.} For generator $Q$ and propagator $K(h)=e^{hQ}$,
small-$h$ expansions give $D_{\mathrm{KL}}(p\Vert pK(h))=\tfrac{h^2}{2}\,\dot{\mathcal{I}}(p;Q)+O(h^3)\ge0$,
while $V(p(t+h))\le V(p(t))$ for all $h>0$ (Theorem~\ref{thm:ctmc-primitive} ensures positivity/primitivity).
Integrating yields the continuous-time analogue:
\[
\int_0^t \lim_{h\downarrow0}\frac{1}{h}D_{\mathrm{KL}}(p(s)\Vert p(s)K(h))\,ds\;+\;V(p(0)) - V(p(t))\ \ge\ 0.
\]

\paragraph{Summary.}
We (i) defined $D_n$ and $\Delta V_n$, (ii) separated their sources of nonnegativity (KLD vs.\ Lyapunov monotonicity), (iii) derived the telescoping-form inequality, and (iv) characterized equality, interpretation, and extensions.
This mirrors Hatano--Sasa/Esposito--Van den Broeck–type formulations for Markov jump processes while crucially referencing a \emph{set} of reachable steady states rather than a single invariant measure.

\section{Discussion}\label{sec:discussion}

\paragraph{KLD potential as informational free energy.}
Under block invariance, the coarse variables $\{w_j\}$ are conserved (Lemma~\ref{lem:block-preservation}) while the within-block conditionals relax toward block-wise invariants. 
The potential $V$ in \eqref{eq:V-def-intro} therefore measures the nonequilibrium information retained beyond the conserved coarse structure. 
Its monotone decay \eqref{eq:V-monotone} provides an information-theoretic second law for non-ergodic replication: fine-scale distinguishability is irreversibly lost, whereas coarse composition is preserved.

\paragraph{Per-step bookkeeping and physical units.}
The minimal second law above decomposes the step production into two nonnegative parts,
$D_n$ and $\Delta V_n$, cf. (45). 
In physical units, the potential drop defines a \emph{informational free-energy} change.
\begin{equation}
\Delta F_{\mathrm{info}}^{(n)} := k_{\mathrm B}T\,\Delta V_n
= k_{\mathrm B}T\,[\,V(p_n)-V(p_{n+1})\,] \ \ge 0,
\label{eq:DeltaF-info}
\end{equation}
while $k_{\mathrm B}T\,D_n$ quantifies the stepwise dissipated ``informational heat.'' 
Together they account for the net degradation incurred by one replication step.

\paragraph{Ergodic vs.\ non-ergodic limits.}
If the reachable steady set $\Pi(p_0)$ collapses to a single invariant distribution (ergodic limit), $V$ reduces to the standard KL-to-steady Lyapunov function. 
Under non-ergodic block-invariant constraints, $V$ is \emph{minimal} KLD to the steady \emph{set} $\Pi(p_0)$, so its decay quantifies relaxation under conserved coarse masses. 
This distinction is essential for replication processes where structural constraints inhibit global mixing.

\paragraph{Tightness and equality cases.}
The equality in \eqref{eq:V-monotone} requires blockwise tightness of the data processing inequality (Remark~\ref{rem:strictness-expanded}); operationally, this corresponds to being in (or effectively in) a blockwise steady state. 
Despite such fixed points, the decrease is strict and $V$ serves as a sensitive progress variable for replication-driven relaxation.

\paragraph{Thermodynamic meaning of a non-zero asymptotic $V$.}
\noindent\textit{Interpretation at long times.}~%
The following clarifies the thermodynamic reading of a non-zero asymptotic $V$.
Our KLD potential $V$ quantifies distance to the \emph{reachable steady set} that preserves coarse variables (block masses). In a driven NESS with primitive blocks, $V(p_n)\!\to\!0$ even though steady dissipation (housekeeping heat) can remain positive; $V$ is a state function, not a direct measure of instantaneous entropy production. A strictly positive long-time value emerges only when the trajectory is persistently kept away from the reachable set, e.g.\ due to weak inter-block leakage, non-primitive blocks, or active enforcement of coarse masses. In that case, choosing $\pi_j=p^{(j)}$ yields
\[
V(p)=\sum_j w_j(p)\,\log\!\frac{w_j(p)}{w_j(p_0)},
\]
so $k_{\mathrm B}T\,V(p)$ is the minimal informational free-energy cost to \emph{refresh} the state back to the admissible coarse constraint. If refreshes occur with frequency $f$, the corresponding maintenance power satisfies $\mathcal P_{\min}\ge f\,k_{\mathrm B}T\,V_{\!*}$, where $V_{\!*}$ is the stationary value under the leak/drive. Thus, a non-zero asymptotic $V$ can be read as a lower bound on the energetic cost of maintaining a coarse-grained steady state while microscopic erasure/copying continuously occurs~\cite{HatanoSasa2001,EspositoVandenBroeck2010}.

\section{Conclusion and Outlook}\label{sec:conclusion}

We formulate replication as a discrete Markov map under block invariance and prove that \emph{minimal} KLD to the reachable steady set,
\[
V(p)=\inf_{\pi\in\Pi(p_0)}D_{\mathrm{KL}}(p\|\pi),
\]
is a Lyapunov function (Theorem~\ref{thm:V-Lyapunov}). 
This yields an information-theoretic second law for \emph{non-ergodic} replication, complementary to Landauer-type results for erasure. 
Instantiations in Gaussian image copying and DNA block-diagonal substitution with proofreading exhibit the predicted behavior: monotone $V$, increasing $H$ and cross-entropy, and decreasing $D_{\mathrm{KL}}(p_n\|q_n)$. 
Future directions include (i) force-resolved single-molecule tests of potential drops and no-free-copying bounds, (ii) extensions to multiscale block hierarchies and heterogeneous kinetics, and (iii) quantum/continuous-state analogues where contractive metrics (e.g.\ quantum $f$-divergences) may provide replication second laws under conserved coarse observables.

\FloatBarrier

\section*{Appendix A: Derivation and leakage threshold}
\label{app:robustness-derivation}

\paragraph{Closed-form expression for $V_\delta$ (corrected).}
Write $\pi=\sum_{j=1}^m w_j\,\pi_j$ with $w\in[0,1]^m$, $\sum_{j=1}^m w_j=1$ and $\pi_j\in\Delta(\mathcal{X}_j)$.
By the block decomposition of KL,
\[
D_{\mathrm{KL}}(p\Vert \pi)
=\sum_{j=1}^m w_j(p)\,D_{\mathrm{KL}}\!\bigl(p^{(j)}\Vert \pi_j\bigr)
 +\sum_{j=1}^m w_j(p)\,\log\!\frac{w_j(p)}{w_j}.
\]
For fixed $w$, the inner minimization is attained at $\pi_j=p^{(j)}$. Thus
\[
V_\delta(p)
=\min_{\substack{w:\ \sum_{j=1}^m w_j=1\\ a_j\le w_j\le b_j}}
\;\varphi(w),\qquad
\varphi(w):=\sum_{j=1}^m w_j(p)\,\log\!\frac{w_j(p)}{w_j},
\]
with $a_j:=w_j(p_0)-\delta$, $b_j:=w_j(p_0)+\delta$. This is a convex program. The Lagrangian
with multipliers $\tau\in\mathbb{R}$ (simplex) and $\mu_j^\pm\ge0$ (box) is
\[
\mathcal{L}(w,\tau,\mu^\pm)=
\sum_{j=1}^m w_j(p)\,\log\frac{w_j(p)}{w_j}
+\tau\!\left(\sum_{j=1}^m w_j-1\right)
+\sum_{j=1}^m \mu_j^+(w_j-b_j)+\sum_{j=1}^m \mu_j^-(a_j-w_j).
\]
KKT conditions (necessary and sufficient here) yield, for an optimum $(w^\star,\tau^\star,\mu^{\pm\star})$,
\[
-\frac{w_j(p)}{w_j^\star}+\tau^\star+\mu_j^{+\star}-\mu_j^{-\star}=0,\quad
\sum_{j=1}^m w_j^\star=1,\quad a_j\le w_j^\star\le b_j,
\]
\[
\mu_j^{+\star}(w_j^\star-b_j)=\mu_j^{-\star}(a_j-w_j^\star)=0,\qquad \mu_j^{\pm\star}\ge0.
\]
Hence interior coordinates ($a_j<w_j^\star<b_j$) satisfy $w_j^\star=w_j(p)/\tau^\star$; active coordinates
stick to the nearest bound. Equivalently, the solution is the \emph{scaled-and-clipped} vector
that also satisfies the simplex constraint:
\[
w_j^\star=\Bigl[\frac{w_j(p)}{\tau^\star}\Bigr]_{[\,a_j,\,b_j\,]},
\qquad
\sum_{j=1}^m w_j^\star=1,
\]
where $\tau^\star>0$ is uniquely determined by the normalization. Substituting back gives
\begin{equation}
\boxed{\;
V_\delta(p)
=\sum_{j=1}^{m} w_j(p)\,
\log\!\frac{w_j(p)}{\,w_j^\star\,},
\quad
w_j^\star=\Bigl[\frac{w_j(p)}{\tau^\star}\Bigr]_{[\,w_j(p_0)-\delta,\;w_j(p_0)+\delta\,]},
\ \ \sum_j w_j^\star=1.
\;}
\label{eq:Vdelta-closed-appendix}
\end{equation}
The associated minimizer is $\pi_j^\star=p^{(j)}$ on each block, i.e.\
$\pi^\star\!\restriction_{\mathcal{X}_j}=w_j^\star\,p^{(j)}$.

\paragraph{Continuity.}
As $\delta\to 0$, the box $[a_j,b_j]$ collapses to $w_j(p_0)$, so $w^\star\to w(p_0)$ and
$V_\delta(p)\to V(p)$.

\paragraph{Sufficient threshold for strict decrease under leakage.}
Let $w_n=(w_1(p_n),\dots,w_m(p_n))$ and $\Delta w:=w_{n+1}-w_n$. Define
\[
\kappa_n
:=\sum_{j=1}^m w_j(p_0)\Bigl[
D_{\mathrm{KL}}\!\bigl(p_n^{(j)}\Vert \pi_j^\ast\bigr)
-D_{\mathrm{KL}}\!\bigl(T_j p_n^{(j)}\Vert \pi_j^\ast\bigr)
\Bigr]\ \ge 0,
\qquad
w_{\min}:=\min_{1\le j\le m} w_j(p_0).
\]
Using the KL block decomposition between steps,
\[
\begin{aligned}
V(p_{n+1})-V(p_n)
&=\sum_{j=1}^m\!\Bigl[w_j(p_{n+1})-w_j(p_n)\Bigr]\!
\log\!\frac{w_j(p_{n+1})}{w_j(p_0)}\\
&\quad+\sum_{j=1}^m w_j(p_0)\!\Bigl[D_{\mathrm{KL}}\!\bigl(T_jp_n^{(j)}\Vert\pi_j^\ast\bigr)
-D_{\mathrm{KL}}\!\bigl(p_n^{(j)}\Vert\pi_j^\ast\bigr)\Bigr],
\end{aligned}
\]
the second line is $-\kappa_n$. By convexity of $x\mapsto x\log x$ and $w_j\ge w_{\min}$,
the first line is bounded by $\tfrac{1}{2w_{\min}}\lVert\Delta w\rVert_1^2$. Therefore
\[
V(p_{n+1})-V(p_n)\ \le\ \frac{1}{2w_{\min}}\lVert\Delta w\rVert_1^2-\kappa_n,
\]
and a sufficient condition for a strict one-step decrease is
\[
\boxed{\;
\lVert\Delta w\rVert_1<2\sqrt{\,w_{\min}\,\kappa_n\,}\ \Rightarrow\ V(p_{n+1})<V(p_n)\!.
\;}
\]

\section*{Appendix B:Simulation protocol for Fig.~1 (Gaussian-copy model)}
\label{app:sim-gaussian}

\noindent\textbf{Software and environment.}
All results in Fig.~1 were generated in Python (NumPy/SciPy/Matplotlib).
Gaussian smoothing uses \texttt{scipy.ndimage.gaussian\_filter} with
boundary condition \texttt{mode="reflect"}.

\medskip
\noindent\textbf{Domain and initialization.}
We use three types of pixel grid. The image domain is partitioned into a
regular $B_x\times B_y$ tiling of rectangular blocks (Fig.~1 uses
$B_x=B_y=4$; we also verified $2\times2$, $4\times4$ and $128\times128$).
Inside each block, a binary pattern is drawn i.i.d. from a Bernoulli law
with a checkerboard success probability:
even $(i{+}j)$ blocks use $0.8$, odd $(i{+}j)$ blocks use $0.2$.
The random seed is fixed to $7$. The resulting matrix is normalized to a
probability mass function (pmf)
$p_0=I_0/\!\sum_{x,y} I_0(x,y)$.

\medskip
\noindent\textbf{Replication step ($n\!\to\!n{+}1$).}
Let $p_n$ denote the current pmf on pixels. We set $\sigma=1.5$ pixels and
apply a Gaussian smoothing, followed by renormalization so that
$\sum_{x,y} q_n(x,y)=1$. Two cases are considered:
\begin{itemize}
\item \emph{Ergodic (global)}:
Gaussian smoothing is applied to the entire image to produce $q_n$.
\item \emph{Non-ergodic (blockwise)}:
The image is split into blocks; the same Gaussian is applied \emph{independently}
within each block (no cross-block smoothing). The blocks are then stitched
back together to form $q_n$. This preserves the total mass inside each block
at every step.
\end{itemize}
In both cases, we update $p_{n+1}=q_n$. We run $n_{\text{steps}}=50$ and
store snapshots at $n\in\{0,10,20,30,40,50\}$.

\medskip
\noindent\textbf{Information measures (base $=2$, in bits).}
At each step we compute
\begin{align}
H(q_n) &:= -\sum_{x,y} q_n \log_2 q_n,\\
H(p_n,q_n) &:= -\sum_{x,y} p_n \log_2 q_n,\\
D_{\mathrm{KL}}(p_n\|q_n) &:= \sum_{x,y} p_n \log_2\!\frac{p_n}{q_n},
\end{align}
where terms with $p_n=0$ or $q_n=0$ are omitted in the sum to avoid
$\log 0$. (To convert bits to nats multiply by $\ln 2$.)

\medskip
\noindent\textbf{KLD potential $V(p_n)$.}
We use the Lyapunov-type potential consistent with the main text:
\begin{itemize}
\item \emph{Ergodic (global) case:}
$V(p_n)=D_{\mathrm{KL}}(p_n\|\pi^\ast)$, where $\pi^\ast$ is the unique invariant histogram under global convolution dynamics (approximated numerically by iterating $T_\sigma$ to stationarity).
\item \emph{Non-ergodic (blockwise) case:}
Let $b$ index blocks, $w_b=\sum_{(x,y)\in b} p_n(x,y)$, $p_b$ the normalized histogram within the block, and $u_b$ the uniform histogram over block $b$. Then
\[
V(p_n)=\sum_{b} w_b\, D_{\mathrm{KL}}\!\bigl(p_b\|u_b\bigr),
\]
which coincides with the distance to the reachable steady set with fixed block masses.
\end{itemize}

\medskip
\noindent\textbf{Visualization and units.}
Time series plots show $H$, $H(p_n,q_n)$, $D_{\mathrm{KL}}(p_n\|q_n)$, and
$V(p_n)$ versus the replication step $n$. Snapshots at the specified steps
use a common grayscale colorbar. All values are reported in \textbf{bits}.

\medskip
\noindent\textbf{Sanity condition on $\sigma$.}
If the Gaussian width exceeds approximately one-quarter of the smallest block edge,
boundary effects may dominate. In Fig.~1 (block size $128{\times}128$,
$\sigma=1.5$) this issue does not arise. The same qualitative trends were
observed for $2{\times}2$ and $4{\times}4$ partitions.

\section*{Appendix C:Simulation details for Fig.~2 (time series)}

\paragraph*{Model reference \& parameters (Fig.~2).}
We use the two-block Markov model of App.~C.2
(block-diagonal kernel; one-step mixture $\widetilde T=(1-\rho)T+\rho\mathcal R$;
effective rates as in Eq.~\eqref{eq:eff-rates} and blockwise invariants $\pi_j^\ast$ therein).
For Fig.~2 we set (nats)
\[
\alpha  = 0.020,\ \beta  = 0.010,\ \gamma = \delta = 0.015,\quad
\alpha' = 0.005,\ \beta' = 0.003,\ \gamma' = \delta' = 0.004,\quad
\rho = 0.30,
\]
giving
\[
\alpha_{\! e}=0.0155,\ \beta_{\! e}=0.0079,\ 
\gamma_{\! e}=\delta_{\! e}=0.0117,
\]
and invariants
\[
\pi_1^\ast=\Bigl(\tfrac{\beta_{\! e}}{\alpha_{\! e}+\beta_{\! e}},\,
\tfrac{\alpha_{\! e}}{\alpha_{\! e}+\beta_{\! e}}\Bigr)\approx(0.3376,\,0.6624),\qquad
\pi_2^\ast=(0.5,\,0.5).
\]

\paragraph*{Initial condition and conserved masses.}
We initialize $p_0=(p_A,p_T,p_C,p_G)=(0.6,0.1,0.2,0.1)$ and normalize.
Block masses
\[
w_1(p):=p_A+p_T,\qquad w_2(p):=p_C+p_G
\]
are conserved by block invariance; we store $w_1(p_0)$ and $w_2(p_0)$ for the potential below.

\paragraph*{Recorded quantities (natural logs; nats).}
At step $n$ we set $q_n:=p_n\,\widetilde T$ (for $n<N$) and record
\[
H(q_n)=-\!\sum_x q_n(x)\log q_n(x),\quad
H_\times(n)=-\!\sum_x p_n(x)\log q_n(x),
\]
\[
D_{\mathrm{KL}}(p_n\|q_n)=\sum_x p_n(x)\log\!\frac{p_n(x)}{q_n(x)}.
\]
The KLD potential (Def.~\eqref{eq:V-def-intro}) reduces to the sum of blocks.
\[
V(p_n)=w_1(p_0)\,D_{\mathrm{KL}}\!\bigl(p_n^{(1)}\|\pi_1^\ast\bigr)
+w_2(p_0)\,D_{\mathrm{KL}}\!\bigl(p_n^{(2)}\|\pi_2^\ast\bigr),
\]
with $p_n^{(1)}=(p_A,p_T)/w_1(p_n)$ and $p_n^{(2)}=(p_C,p_G)/w_2(p_n)$.

\paragraph*{Iteration.}
We iterate
\[
p_{n+1}=q_n=p_n\,\widetilde T,\qquad n=0,1,\dots,N-1,
\]
with $N=50$. At each $n$ we first record $V(p_n)$; then, if $n<N$, we compute and
record $H(q_n)$, $H_\times(n)$, and $D_{\mathrm{KL}}(p_n\|q_n)$, and finally set $p_{n+1}=q_n$.

\section*{Appendix D: Simulation details for Fig.~3 (potential landscape)}

\paragraph*{D.1\quad Model and objective.}
We consider two invariant blocks $\mathcal{X}_1=\{\mathrm{A},\mathrm{T}\}$ and
$\mathcal{X}_2=\{\mathrm{C},\mathrm{G}\}$.  On the grid $(x,y)\in[0,1]^2$,
we parameterize a block-wise distribution.
\[
p(x,y)=\bigl(w_1 x,\; w_1(1-x),\; w_2 y,\; w_2(1-y)\bigr),
\]
where $x$ is the A fraction in the AT block and $y$ is the C fraction in the GC block.
The KLD potential (Def.~\eqref{eq:V-def-intro}) reduces to the separable sum
\[
V\bigl(p(x,y)\bigr)=
w_1\,D_{\mathrm{KL}}\!\bigl([x,1-x]\Vert \pi_1^\ast\bigr)\;+\;
w_2\,D_{\mathrm{KL}}\!\bigl([y,1-y]\Vert \pi_2^\ast\bigr),
\]
in nats.  We plot $V$ over $[0,1]^2$ and mark the minimum
$(x^\ast,y^\ast)=\bigl(\pi_1^\ast(\mathrm{A}),\,\pi_2^\ast(\mathrm{C})\bigr)$.

\paragraph*{D.2\quad Effective rates and steady states.}
One step is a convex mixture of an extension channel $T$ and a proofreading
channel $\mathcal{R}$ with probability $\rho=0.30$.  Each $2\times2$ block is
\[
T_1=\begin{pmatrix}1-\alpha&\alpha\\ \beta&1-\beta\end{pmatrix},\quad
R_1=\begin{pmatrix}1-\alpha'&\alpha'\\ \beta'&1-\beta'\end{pmatrix},
\]
\[
T_2=\begin{pmatrix}1-\gamma&\gamma\\ \delta&1-\delta\end{pmatrix},\quad
R_2=\begin{pmatrix}1-\gamma'&\gamma'\\ \delta'&1-\delta'\end{pmatrix},
\]
so, the effective rates are
\begin{equation}
\begin{aligned}
\alpha_{\!e}&=(1-\rho)\alpha+\rho\alpha', &
\beta_{\!e}&=(1-\rho)\beta+\rho\beta',\\
\gamma_{\!e}&=(1-\rho)\gamma+\rho\gamma', &
\delta_{\!e}&=(1-\rho)\delta+\rho\delta'.
\end{aligned}
\label{eq:eff-rates}
\end{equation}

\emph{AT block (Fig.~2):}
$\alpha=0.020,\ \beta=0.010,\ \alpha'=0.005,\ \beta'=0.003$
$\Rightarrow\ \alpha_{\!e}=0.0155,\ \beta_{\!e}=0.0079$.
\emph{GC block:}
$\gamma=0.014,\ \delta=0.021,\ \gamma'=0.004,\ \delta'=0.006$
$\Rightarrow\ \gamma_{\!e}=0.0110,\ \delta_{\!e}=0.0165$.
The blockwise invariants (for $[[1-a,a],[b,1-b]]$) are
\[
\pi_1^\ast=\Bigl(\tfrac{\beta_{\!e}}{\alpha_{\!e}+\beta_{\!e}},\;
\tfrac{\alpha_{\!e}}{\alpha_{\!e}+\beta_{\!e}}\Bigr)
\approx(0.338,\,0.662),
\]
\[
\pi_2^\ast=\Bigl(\tfrac{\delta_{\!e}}{\gamma_{\!e}+\delta_{\!e}},\;
\tfrac{\gamma_{\!e}}{\gamma_{\!e}+\delta_{\!e}}\Bigr)
=(0.60,\,0.40).
\]
We take $(w_1,w_2)=(0.5,0.5)$ unless otherwise noted.  The blue marker in
Fig.~3 is placed at $(x^\ast,y^\ast)=(\pi_1^\ast(\mathrm{A}),\pi_2^\ast(\mathrm{C}))
\approx(0.338,0.600)$.

\paragraph*{D.3\quad Grid and evaluation.}
We use a $101\times101$ grid on $[0,1]^2$.  For Bernoulli pairs we evaluate
\[
D_{\mathrm{KL}}\!\bigl([p,1-p]\Vert[q,1-q]\bigr)
= p\log\tfrac{p}{q}+(1-p)\log\tfrac{1-p}{1-q},
\]
with safe clipping to avoid $\log 0$.  The color scale shows $V$ in nats
(``viridis'' colormap).  Axes: $x$ = A fraction (AT), $y$ = C fraction (GC).

\paragraph*{D.4\quad Minimal Python excerpt (reproducibility).}

\begin{verbatim}
import numpy as np

# mixture rate
rho = 0.30

# AT block (same as Fig.2)
alpha, beta   = 0.020, 0.010
alpha_p, beta_p = 0.005, 0.003
ae = (1 - rho)*alpha + rho*alpha_p      # 0.0155
be = (1 - rho)*beta  + rho*beta_p       # 0.0079

# GC block (asymmetric, to place y* \approx 0.60)
gamma, delta   = 0.014, 0.021
gamma_p, delta_p = 0.004, 0.006
ge = (1 - rho)*gamma + rho*gamma_p      # 0.0110
de = (1 - rho)*delta + rho*delta_p      # 0.0165

# block weights
w1, w2 = 0.5, 0.5

# invariants for [[1-a, a],[b, 1-b]]
pi1_A = be/(ae+be);  pi1_T = ae/(ae+be)     # $\approx$ (0.338, 0.662)
pi2_C = de/(ge+de);  pi2_G = ge/(ge+de)     #   (0.600, 0.400)

def kl_bernoulli(p, q, eps=1e-12):
    p = np.clip(p, eps, 1.0 - eps)
    q = np.clip(q, eps, 1.0 - eps)
    return p*np.log(p/q) + (1-p)*np.log((1-p)/(1-q))

# grid and potential values (nats)
N = 101
xs = np.linspace(0.0, 1.0, N)   # A-fraction in AT
ys = np.linspace(0.0, 1.0, N)   # C-fraction in GC
X, Y = np.meshgrid(xs, ys, indexing='xy')
V = w1*kl_bernoulli(X, pi1_A) + w2*kl_bernoulli(Y, pi2_C)

# minimum location for the blue marker:
x_star, y_star = pi1_A, pi2_C
\end{verbatim}

\noindent
This excerpt computes the potential field $V(p)$ on the grid and the minimum
$(x^\ast,y^\ast)$.  The plotted Fig.~3 uses a filled contour of $V$, the
``viridis'' colormap, and annotates $(x^\ast,y^\ast)$ with a blue marker.

\section*{Appendix E:Block primitivity in DNA Example and convergence rate}
\label{app:block-primitive}
\paragraph{Necessary and sufficient conditions for two-state blocks.}
For $T=\begin{psmallmatrix}1-a & a\\ b & 1-b\end{psmallmatrix}$, 
primitive $\iff a>0,\ b>0$ and $(a,b)\neq(1,1)$. 
Hence $a,b\in(0,1)$ suffices. 
\paragraph{Effect of proofreading/repair mixing.}
If $\mathcal R$ has strictly positive entries and $\rho>0$, then all effective entries of $\widetilde T=(1-\rho)T+\rho\mathcal R$ are strictly positive, implying primitivity even when $T$ is period-2 ($a=b=1$).
\paragraph{Relation to Eq.~\eqref{eq:proofreading-affinity-rewrite} and detailed balance.}
Eq.~\eqref{eq:proofreading-affinity-rewrite} encodes local detailed balance on the enlarged network; the coarse two-state block remains reversible with invariant $\pi^\ast$, so Eq.~(35) sets steady odds and speeds but is not the primitivity criterion.
\paragraph{Convergence rate (sketch).}
The second eigenvalue of a two-state block is $\lambda_2=1-(a+b)$; thus total variation decays geometrically with factor $|\lambda_2|$.

\section*{Appendix F: A minimal non–block–invariant counterexample and a leakage threshold}
\label{app:leak-threshold}

\paragraph{Symbols and definitions.}
State space and blocks:
\[
\mathcal{X}=\{a_1,a_2,b\},\qquad
\mathcal{X}_1=\{a_1,a_2\},\qquad
\mathcal{X}_2=\{b\}.
\]
Initial distribution \(p_0(a_1,a_2,b)=(0.6,0.3,0.1)\).
Block (coarse) masses:
\[
w_1(p)=p(a_1)+p(a_2),\qquad
w_2(p)=p(b).
\]

\paragraph{Reachable steady set (coarse masses fixed).}
\[
\Pi(p_0)=\Bigl\{\sum_{j=1}^m w_j(p_0)\,\pi_j:\ \pi_j\in\mathcal{I}_j\Bigr\}.
\]

\paragraph{One-step update with a non–block–invariant kernel.}
\[
T=
\begin{bmatrix}
0.7&0.2&0.1\\
0.1&0.7&0.2\\
0&0&1
\end{bmatrix},
\qquad
p_1=p_0\,T.
\]
Cross–block entries are \(T_{a_1\to b}=0.1\) and \(T_{a_2\to b}=0.2\).
Define the net leak
\[
\mathcal{L}:=w_2(p_1)-w_2(p_0),\qquad
w_1(p_1)=p_1(a_1)+p_1(a_2),\ \ w_2(p_1)=p_1(b).
\]
(Numerical values can be read off directly from \(p_1\).)

\paragraph{Block KL decomposition and consequence.}
For any \(\pi=\sum_{j} w_j(p_0)\,\pi_j \in \Pi(p_0)\),
\[
\begin{aligned}
D_{\mathrm{KL}}(p\Vert \pi)
&=\sum_{j} w_j(p)\,D_{\mathrm{KL}}\!\bigl(p^{(j)}\Vert \pi_j\bigr)
 \;+\;\sum_{j} w_j(p)\,\log\!\frac{w_j(p)}{w_j(p_0)}.
\end{aligned}
\]
Since \(\pi_1\) is free, the minimizer sets \(\pi_1=p^{(1)}\); for the singleton block,
\(p^{(2)}=\pi_2=\delta_b\). Therefore
\[
V(p)=\inf_{\pi\in \Pi(p_0)} D_{\mathrm{KL}}(p\Vert \pi)
=\sum_{j} w_j(p)\,\log\!\frac{w_j(p)}{w_j(p_0)}.
\]
(Thus \(V\) measures the coarse-mass mismatch when leakage is present.)

\paragraph{A sufficient leakage threshold for one-step decrease.}
Let \(w_n=(w_1(p_n),\dots,w_m(p_n))\) and \(\Delta w:=w_{n+1}-w_n\). Define
\[
\kappa_n
:=\sum_{j} w_j(p_0)\Bigl[
D_{\mathrm{KL}}\!\bigl(p_n^{(j)}\Vert \pi_j^\ast\bigr)
-D_{\mathrm{KL}}\!\bigl(T_j p_n^{(j)}\Vert \pi_j^\ast\bigr)
\Bigr]\ \ge 0,
\qquad
w_{\min}:=\min_{j} w_j(p_0).
\]
A convenient sufficient condition for a strict one–step decrease is
\begin{equation}
\boxed{\;
\lVert\Delta w\rVert_1
\;<\; 2\,\sqrt{\,w_{\min}\,\kappa_n\,}
\ \Rightarrow\
V(p_{n+1})<V(p_n)\!.
\;}
\label{eq:leak-threshold}
\end{equation}

\paragraph{Proof sketch / origin of the constants.}
By the block KL decomposition,
\[
\begin{aligned}
V(p_{n+1})-V(p_n)
&=\sum_{j} \Bigl[w_j(p_{n+1})\log\!\frac{w_j(p_{n+1})}{w_j(p_0)}
-w_j(p_{n})\log\!\frac{w_j(p_{n})}{w_j(p_0)}\Bigr]\\
&=\underbrace{\sum_{j}\Delta w_j\,\log\!\frac{w_j(p_n)}{w_j(p_0)}}_{\text{linear term}}
+\underbrace{\sum_{j}\bigl[w_j(p_{n+1})-w_j(p_{n})\bigr]\log\!\frac{w_j(p_{n+1})}{w_j(p_n)}}_{\text{quadratic term}}.
\end{aligned}
\]
The quadratic term is upper bounded by \(\tfrac{1}{2w_{\min}}\lVert\Delta w\rVert_1^2\) via convexity of \(x\mapsto x\log x\) and \(w_j\ge w_{\min}\).
The linear term is at most \(-\,\kappa_n\) by definition of \(\kappa_n\) (within-block contraction toward \(\pi_j^\ast\)).
Combining the two gives
\[
V(p_{n+1})-V(p_n)\ \le\ \frac{1}{2w_{\min}}\lVert\Delta w\rVert_1^2\;-\;\kappa_n,
\]
which yields the threshold \eqref{eq:leak-threshold}.

\section*{Appendix G: Non-primitive block with a changing minimizer}
\label{app:changing-minimizer}

\noindent\emph{Notation.} Here $K$ denotes a \textbf{within-block} abstract Markov kernel (not the image-induced $T_\sigma$ of Sec.~\ref{subsec:image}).

Consider a single block $\mathcal X_j=\{u,v,t\}$ with a \emph{reducible} kernel
\begin{equation}
K\;=\;
\begin{pmatrix}
1 & 0 & 0\\
0 & 1 & 0\\
\varepsilon & 1-\varepsilon & 0
\end{pmatrix},
\qquad \varepsilon\in(0,1).
\label{eq:F-reducible-K}
\end{equation}
States $u$ and $v$ are absorbing; $t$ is transient and jumps to $u$ or
$v$ in one step with probabilities $\varepsilon$ and $1-\varepsilon$,
respectively. Hence the invariant set is
\begin{equation}
\mathcal I_j
= \bigl\{\,\lambda\,\delta_u+(1-\lambda)\,\delta_v:\ \lambda\in[0,1]\,\bigr\}.
\label{eq:F-invariant-set}
\end{equation}
Let $p^{(j)}_n$ denote the within-block law started from
$p^{(j)}_0=(0,0,1)$. Then
\begin{equation}
p^{(j)}_1=(\varepsilon,1-\varepsilon,0),\qquad
p^{(j)}_n=p^{(j)}_1\ \ \text{for all }n\ge1.
\label{eq:F-traject-const}
\end{equation}
For any $\lambda\in[0,1]$,
\begin{equation}
D_{\mathrm{KL}}\!\bigl(p^{(j)}_1\;\Vert\;\lambda\,\delta_u+(1-\lambda)\,\delta_v\bigr)
=\varepsilon\log\tfrac{\varepsilon}{\lambda}
+(1-\varepsilon)\log\tfrac{1-\varepsilon}{1-\lambda},
\label{eq:F-Dkl-lambda}
\end{equation}
which is minimized at $\lambda=\varepsilon$.

\paragraph*{Nonincrease of $V$ and minimizer switching.}
If the global partition weight on this block is $w_j(p_n)$ and other blocks are fixed,
the contribution of block $j$ to $V$ is $w_j(p_0)\,D_{\mathrm{KL}}\!\bigl(p_n^{(j)}\Vert\pi_j\bigr)$
for some $\pi_j\in\mathcal I_j$. From \eqref{eq:F-traject-const}–\eqref{eq:F-Dkl-lambda},
for $n\ge1$ the minimizer is $\pi_j^\star=\varepsilon\,\delta_u+(1-\varepsilon)\,\delta_v$ and
\[
V\ \text{is constant for }n\ge 1 \ \text{(hence nonincreasing)}.
\]
If we perturb the initial condition to
$p^{(j)}_0=(0,0,1-\eta)+(\eta,0,0)$ with a small $\eta>0$, then at time $0$ the
minimizer is $\lambda^\star_0>\varepsilon$, whereas for $n\ge1$ it becomes
$\lambda^\star_1=\varepsilon$. Thus the blockwise minimizer \emph{changes between steps},
yet $V$ remains nonincreasing by Theorem~\ref{thm:V-Lyapunov}. This shows that
minimizer switching is generic whenever $\mathcal I_j$ is not a singleton (reducible blocks).


\begin{thebibliography}{99}

\bibitem{HershbergPetrov2010}
R. Hershberg and D. A. Petrov.
Evidence That Mutation Is Universally Biased Towards AT in Bacteria.
\emph{PLoS Genet.}, \textbf{6}(9):e1001115, 2010.
\href{https://doi.org/10.1371/journal.pgen.1001115}{doi: 10.1371/journal.pgen.1001115}.

\bibitem{FryxellMoon2005}
K. J. Fryxell and W. J. Moon.
CpG Mutation Rates in the Human Genome Are Highly Dependent on Local GC Content.
\emph{Mol. Biol. Evol.}, \textbf{22}(3):650--658, 2005.
\href{https://doi.org/10.1093/molbev/msi043}{doi: 10.1093/molbev/msi043}.

\bibitem{Cooper2010}
D. N. Cooper, M. Mort, P. D. Stenson, E. V. Ball, and N. A. Chuzhanova.
Methylation-mediated deamination of 5-methylcytosine appears to give rise to mutations causing human inherited disease in CpNpG trinucleotides, as well as in CpG dinucleotides.
\emph{Hum. Genomics}, \textbf{4}(6):406, 2010.
\href{https://doi.org/10.1186/1479-7364-4-6-406}{doi: 10.1186/1479-7364-4-6-406}.

\bibitem{AggarwalaVoight2016}
V. Aggarwala and B. F. Voight.
An expanded sequence context model broadly explains variability in polymorphism levels across the human genome.
\emph{Nat. Genet.}, \textbf{48}(4):349--355, 2016.
\href{https://doi.org/10.1038/ng.3511}{doi: 10.1038/ng.3511}.

\bibitem{DuretGaltier2009}
L. Duret and N. Galtier.
Biased gene conversion and the evolution of mammalian genomic landscapes.
\emph{Annu. Rev. Genom. Hum. Genet.}, \textbf{10}:285--311, 2009.
\href{https://doi.org/10.1146/annurev-genom-082908-150001}{doi: 10.1146/annurev-genom-082908-150001}.

\bibitem{CapraPollard2013}
J. A. Capra, M. J. Hubisz, D. Kostka, K. S. Pollard, and A. Siepel.
A Model-Based Analysis of GC-Biased Gene Conversion in the Human and Chimpanzee Genomes.
\emph{PLoS Genet.}, \textbf{9}(8):e1003684, 2013.
\href{https://doi.org/10.1371/journal.pgen.1003684}{doi: 10.1371/journal.pgen.1003684}.

\bibitem{Weber2014}
C. C. Weber, B. Boussau, J. Romiguier, E. D. Jarvis, and H. Ellegren.
Evidence for GC-biased gene conversion as a driver of between-lineage differences in avian base composition.
\emph{Genome Biol.}, \textbf{15}(12):549, 2014.
\href{https://doi.org/10.1186/s13059-014-0549-1}{doi: 10.1186/s13059-014-0549-1}.

\bibitem{Wuite2000}
G. J. L. Wuite, S. B. Smith, M. Young, D. Keller, and C. Bustamante.
Single-Molecule Studies of the Effect of Template Tension on T7 DNA Polymerase Activity.
\emph{Nature}, \textbf{404}(6773):103--106, 2000.
\href{https://doi.org/10.1038/35003614}{doi: 10.1038/35003614}.

\bibitem{Bustamante2003}
C. Bustamante, Z. Bryant, and S. B. Smith.
Ten years of tension: single-molecule DNA mechanics.
\emph{Nature}, \textbf{421}(6921):423--427, 2003.
\href{https://doi.org/10.1038/nature01405}{doi: 10.1038/nature01405}.

\bibitem{Le2018}
T. T. Le, Y. Yang, C. Tan, M. M. Suhanovsky, R. M. Fulbright Jr., J. T. Inman, M. Li, J. Lee, S. Perelman, J. W. Roberts, A. M. Deaconescu, and M. D. Wang.
Mfd Dynamically Regulates Transcription via a Release and Catch-Up Mechanism.
\emph{Cell}, \textbf{172}(1--2):344--357.e15, 2018.
\href{https://doi.org/10.1016/j.cell.2017.11.017}{doi: 10.1016/j.cell.2017.11.017}.

\bibitem{Gaspard2016a}
P. Gaspard.
Kinetics and Thermodynamics of DNA Polymerases without Exonuclease Activity.
\emph{Phys. Rev. E}, \textbf{93}(4):042419, 2016.
\href{https://doi.org/10.1103/PhysRevE.93.042419}{doi: 10.1103/PhysRevE.93.042419}.

\bibitem{Pineros2020}
W. D. Pi{\~n}eros and T. Tlusty.
Kinetic Proofreading and the Limits of Thermodynamic Uncertainty.
\emph{Phys. Rev. E}, \textbf{101}(2):022415, 2020.
\href{https://doi.org/10.1103/PhysRevE.101.022415}{doi: 10.1103/PhysRevE.101.022415}.

\bibitem{Hopfield1974}
J. J. Hopfield.
Kinetic Proofreading: A New Mechanism for Reducing Errors in Biosynthetic Processes Requiring High Specificity.
\emph{Proc. Natl. Acad. Sci. U.S.A.}, \textbf{71}(10):4135--4139, 1974.
\href{https://doi.org/10.1073/pnas.71.10.4135}{doi: 10.1073/pnas.71.10.4135}.

\bibitem{SartoriPigolotti2015}
P. Sartori and S. Pigolotti.
Thermodynamics of Error Correction.
\emph{Phys. Rev. X}, \textbf{5}(4):041039, 2015.
\href{https://doi.org/10.1103/PhysRevX.5.041039}{doi: 10.1103/PhysRevX.5.041039}.

\bibitem{HatanoSasa2001}
T. Hatano and S.-i. Sasa.
Steady-State Thermodynamics of Langevin Systems.
\emph{Phys. Rev. Lett.}, \textbf{86}(16):3463--3466, 2001.
\href{https://doi.org/10.1103/PhysRevLett.86.3463}{doi: 10.1103/PhysRevLett.86.3463}.

\bibitem{EspositoVandenBroeck2010}
M. Esposito and C. Van den Broeck.
Three Faces of the Second Law. I. Master Equation Formulation.
\emph{Phys. Rev. E}, \textbf{82}(1):011143, 2010.
\href{https://doi.org/10.1103/PhysRevE.82.011143}{doi: 10.1103/PhysRevE.82.011143}.

\end{thebibliography}
\end{document}